\documentclass[aps,superscriptaddress,amsmath, amssymb, amsfonts,twocolumn]{revtex4}
\usepackage{float}
\usepackage[sort&compress]{natbib}
\usepackage{graphicx}
\usepackage{float}

\newcommand{\blue}{}
\newcommand{\red}{}

\newcommand{\magenta}{}
\newcommand{\black}{}

\newtheorem{theorem}{Theorem}

\newtheorem{corollary}{Corollary}

\newtheorem{definition}{Definition}
\newtheorem{example}{Example}

\newtheorem{lemma}{Lemma}

\newtheorem{proposition}{Proposition}

\newenvironment{proof}[1][Proof]{\noindent\textbf{#1.} }{\ $\Box$} 

\newcommand{\Tr}{\mbox{Tr}}

\renewcommand{\FR}[2]{\stackrel{\longleftrightarrow}{#1 #2}}

\renewcommand{\H}{{\bf H}}

\newcommand{\K}{{\bf K}}

\newcommand{\beq}{\begin{equation}}
\newcommand{\eeq}{\end{equation}}
\newcommand{\A}{{\mathfrak A}}
\newcommand{\B}{{\mathfrak B}}
\newcommand{\E}{{\cal E}}
\newcommand{\interior}{\mathrm{int~}}
\newcommand{\cone}{\mathrm{cone}}

\newcommand{\aff}{\mathrm{aff}}
\newcommand{\conv}{\mathrm{conv}}

\newcommand{\lin}{\mathrm{lin}}

\newcommand{\intersect}{\cap}

\begin{document}

\title{Entropy and Information Causality in General Probabilistic Theories}

\date{September 28, 2009}

\author{Howard Barnum} \email{hbarnum@perimeterinstitute.ca} \affiliation{Perimeter Institute for Theoretical Physics, 31 Caroline St N, Waterloo, Ontario, N2L 2Y5 Canada}

\author{Jonathan Barrett} \email{j.barrett@bristol.ac.uk}
\affiliation{H.~H. Wills Physics Laboratory, University of Bristol, Tyndall Avenue,
Bristol, BS8 1TL United Kingdom}

\author{Lisa Orloff Clark}\email{clarklisa@susq.edu}
\affiliation{Department of Mathematical Sciences, Susquehanna University, Selinsgrove, PA, 17870 USA}

\author{Matthew Leifer} \email{matt@mattleifer.info}
\affiliation{Perimeter Institute for Theoretical Physics, 31 Caroline St N, Waterloo, Ontario, N2L 2Y5 Canada}
\affiliation{Institute for Quantum Computing, University of Waterloo, 200 University
Ave. W, Waterloo, Ontario, N2L 3G1 Canada}

\author{Robert Spekkens}\email{rspekkens@perimeterinstitute.ca}
\affiliation{Perimeter Institute for Theoretical Physics, 31 Caroline St N, Waterloo, Ontario, N2L 2Y5 Canada}

\author{Nicholas Stepanik}\email{stepanik@susqu.edu}
\affiliation{Department of Mathematical Sciences, Susquehanna University, Selinsgrove, PA, 17870 USA}

\author{Alex Wilce} \email{wilce@susqu.edu} \affiliation{Department of Mathematical Sciences, Susquehanna University, Selinsgrove, PA, 17870 USA}

\author{Robin Wilke}\email{rwilke@uvm.edu}\affiliation{Department of Mathematics and Statistics, University of Vermont, Burlington, VT 05405 USA}

\begin{abstract}

We investigate the concept of entropy in probabilistic theories more
general than quantum mechanics, with particular reference to the
notion of information causality recently proposed by Pawlowski {\em et
  al.} (arXiv:0905.2992). We consider two entropic quantities, which
we term {\em measurement} and {\em mixing} entropy. In the context of
classical and quantum theory, these coincide, being given by the
Shannon and von Neumann entropies respectively; in general, however,
they are very different. In particular, while measurement entropy is
easily seen to be concave, mixing entropy need not be. In fact, as we
show, mixing entropy is not concave whenever the state space is a
non-simplicial polytope. Thus, the condition that measurement and
mixing entropies coincide is a strong constraint on possible
theories. We call theories with this property {\em monoentropic}.

Measurement entropy is subadditive, but not in general strongly
subadditive. Equivalently, if we define the mutual information between
two systems $A$ and $B$ by the usual formula $I(A:B) = H(A) + H(B) -
H(AB)$ where $H$ denotes the measurement entropy and $AB$ is a
non-signaling composite of $A$ and $B$, then it can happen that
$I(A:BC) < I(A:B)$.  This is relevant to information
causality in the sense of Pawlowski {\em et al.}: we show that any
monoentropic non-signaling theory in which measurement entropy is
strongly subadditive, and also satisfies a version of the Holevo
bound, is informationally causal, and on the other hand we observe
that Popescu-Rohrlich boxes, which violate information causality, also
violate strong subadditivity.  We also explore the interplay between
measurement and mixing entropy and various natural conditions on theories
that arise in quantum axiomatics.
\end{abstract}

\maketitle

\section{Introduction}

One can view quantum mechanics as an extension of the classical
probability calculus, allowing for random variables that are not
simultaneously measurable. In order to gain a clearer understanding of
quantum theory from this perspective, it is useful to contrast it with
various (factitious) alternatives that are neither classical nor
quantum. The best known example of such a ``foil"  probabilistic
theory is probably the theory of ``non-local boxes"
\cite{PopescuRohrlich, Barrett, SB}; but in fact, there is a standard mathematical
framework for such theories, going back to the work of Mackey in the
1950s \cite{Mackey}. Working in this framework, one can show that many
phenomena commonly regarded as characteristically quantum --
no-cloning and no-broadcasting theorems \cite{BBLW06, BBLW07}, the
trade-off between state disturbance and measurement \cite{Barrett},
and the existence and basic properties of entangled states \cite{Klay,
  Klay88, Barrett, BBLW06} -- are actually quite generic features of
all non-classical probabilistic theories satisfying a basic
``non-signaling" constraint. Other quantum phenomena, such as the
possibility of teleportation \cite{BBLW08} or remote steering of
ensembles \cite{BGW}, are more special (and in some sense, more {\em
  classical}), but can still be seen to arise outside the boundaries
of quantum theory.

One might hope to find some reasonably short list of probabilistic or
information-theoretic phenomena that more cleanly separate quantum
theory from other possible non-signaling theories. In a recent paper
\cite{Petal}, Pawlowski {\em et al.} take a step in this direction by
showing that any non-signaling correlation violating the Tsirel'son
bound also violates a qualitative information-theoretic principle they
call {\em information causality} (IC). In essence, this prohibits a form of
``multiplexing" in which one party (Bob) gains the ability to access a
total of more than $m$ bits of information held by another party
(Alice), on the basis of an $m$-bit message from Alice, plus some
shared non-signaling bipartite state. It is also established in
\cite{Petal} that quantum mechanics -- and hence, also classical
probability theory -- satisfies this IC constraint.

In establishing that quantum mechanics satisfies IC, Pawlowski {\em et
  al.}  make use only of standard formal properties of the von Neumann
entropy of joint quantum states. This raises the obvious question of
where their proof breaks down in other contexts (e.g., a PR box) in
which IC fails.  In order to address this question, we develop some of
the basic machinery of entropy, conditional entropy and mutual
information in a very general probabilistic setting --- an
independently interesting problem, which seems not to have received
much previous attention (an exception being the paper \cite{Hein} of
Hein).

We begin by identifying two notions of entropy, which we call {\em
  measurement} and {\em mixing} entropy, and which we
denote respectively by $H(A)$ and $S(A)$, where $A$ is a general
probabilistic model. Briefly: the measurement entropy of a system is
the minimum Shannon entropy of any possible measurement thereon, while
the mixing entropy is the infimum of the Shannon entropies of the
various ways of preparing the system's state as a mixture of pure
states. These coincide classically and in quantum theory, but are
generally quite different animals. For example, measurement entropy is
always subadditive, and is concave; mixing entropy is generally
neither. In fact, in Appendix A, we show that there are {\em always}
violations of concavity of the mixing entropy for any system with a state
space that is a non-simplicial polytope.  Thus, the condition that mixing and
measurement entropies {\em do} coincide, as in quantum mechanics, is a
powerful constraint on the structure of a probabilistic theory. We
call theories with this feature {\em monoentropic}.

Next, we develop an account of joint measurement entropy, conditional entropy, and mutual information for composite systems,
and apply this apparatus to the notion of information causality given in \cite{Petal}.
Somewhat surprisingly, it seems that the main issue is not so much one of the strength of non-local correlations,
but rather, the failure, of two other, very basic principles. One is the strong subadditivity, or, equivalently,
the condition that the mutual information, defined by $I(A:B) = H(A) + H(B) - H(AB)$, satisfy
\[I(A:BC) \geq I(A:C).\]
This holds both classically and in quantum theory, but is violated
in very simple non-classical models -- even models in which $A$
and $B$ are {\em classical}, so that no issue of non-locality can
arise. Another basic principle, equivalent to the Holevo bound, is that $I(E:B) \leq I(A:B)$ where
$E$ is any particular measurement on system $A$.

Both strong subadditivity and the Holevo bound can be viewed as
special cases of an even more basic principle, usually called the {\em
  data processing inequality}. This asserts that, for any systems $A,
B$ and $B'$, and any reasonable process $\E : B \rightarrow B'$, we
have $I(A:\E(B)) \leq I(A:B)$ (where $\E(B) := B'$ is the output system of
the process). This is intuitively appealing as a basic physical
postulate.

Finally, we apply the apparatus just described to the notion of
information causality. We consider in detail the basic example, due to
van Dam \cite{vanDam} of an IC-violating composite system, and find
that it exhibits a violation of strong subadditivity.  We also
establish that, within a very broad class of finite-dimensional {\em
  monoentropic} theories, strong subadditivity together with the
Holevo bound entail information causality. It remains an open question
whether all three of these conditions are necessary for this
conclusion.

The remainder of this paper is organized as follows. In Section II, we
review in some detail the framework of generalized probability theory,
largely following \cite{BBLW06}.  In section III, we define, and
establish some elementary properties of, measurement and mixing
entropy for states of an arbitrary probabilistic model. Section IV
discusses composite systems in our framework, and collects some
observations about the behavior of joint measurement entropy, and the
notion of mutual information based on this.  Using this apparatus, we
establish in Section V that any monoentropic probabilistic theory in
which measurement entropy is strongly subadditive and satisfies the
Holevo bound, is informationally causal in the sense of \cite{Petal}.
We also point out that violations of strong subadditivity are possible
in theories having no entanglement. Section VI collects some final
remarks and open questions.  Appendix A contains the proof that mixing
entropy is not concave on state spaces that are non-simplicial
polytopes.  Appendix B establishes some further properties of
monoentropic theories, relevant to axiomatic characterizations of
quantum theories, and also shows that monoentropicity follows from two
other properties, steering and pure conditioning, the physical content
of which may be more transparent.  Finally in Appendix
\ref{appendix:linearization} we discuss how the framework of this
paper relates to the ``convex sets'' framework, and consider analogous
definitions of measurement entropy in that context.

\section{General Probabilistic Models}\label{testspaces}

As we mentioned above, there is a more or less standard mathematical
framework for discussing general probabilistic models, going back at
least to the work of Mackey in the 1950s, and further developed (or,
in some cases, rediscovered) in succeeding decades by various authors
\cite{DaviesLewis, Edwards, Ludwig, FR, Hardy, Barrett}.  In what
follows, we work in the idiom of \cite{BL}, which we briefly recall.

We characterize a probabilistic model, or, more briefly, a {\em
  system}, by a pair $A = ({\mathfrak A},\Omega)$ where $\mathfrak A$
is a collection -- possibly infinite -- of discrete classical
experiments or measurements, and $\Omega$ is a set of {\em states}. We
make the following assumptions:
\begin{itemize}
\item[(i)] Every experiment in $\mathfrak A$ is defined by its set of possible outcomes, so that
we may represent $\mathfrak A$, mathematically, as a collection of sets $E, F, ...$. In the language of \cite{FR,Wilce09}, this is a {\em test space}; accordingly, we refer to the various sets $E, F, ... \in {\mathfrak A}$ as {\em tests}.
\item[(ii)] Every state is entirely determined by the probabilities
it assigns to the outcomes of the various measurements in $\mathfrak A$. Thus, letting $X := \bigcup {\mathfrak A}$ denote the total {\em outcome space} of $\mathfrak A$, $\Omega$ consists of functions $\alpha : X \rightarrow [0,1]$, with $\sum_{x \in E} \alpha(x) = 1$ for every set $E \in {\mathfrak A}$.
\item[(iii)] The state space $\Omega$ is a convex subset of $[0,1]^{X}$ (the functions from $X$ to $[0,1]$). Hence any statistical mixture of states is a state.
\end{itemize}
For a given test space ${\mathfrak A}$, one can define the space of \emph{all} states on ${\mathfrak A}$.  This is called the \emph{maximal} state space and is denoted by  $\Omega({\mathfrak A})$. It is clearly convex.  The \emph{physical} state space $\Omega$ is necessarily either equal to or a subset of the maximal state space.

This framework, though very simple, is broad enough to accommodate both measure-theoretic classical probability
theory and non-commutative probability theory based on von Neumann algebras.\footnote{Measure-theoretic
classical probability theory is, in effect, the theory of systems of
the form $({\mathfrak D},\Theta)$ where ${\mathfrak D} = {\mathfrak
D}(S,\Sigma)$ is the set of all finite (respectively, countable)
partitions $E = \{a_i\}$ of a measurable space $S$ by non-empty
mesurable sets $a_i \in \Sigma$, and where $\Theta$ is some closed
convex set of probability measures on $E$. The probabilistic
apparatus of states and observables associated with von Neumann
algebras can be modeled in a similar way.}
In this paper, we shall be interested exclusively in discrete, finite-dimensional systems. Accordingly, from this point forward, we make
the standing assumptions that (i)
${\mathfrak A}$ is {\em locally finite}, meaning that all tests $E \in {\mathfrak A}$ are finite sets \footnote{Alternatively, this condition could be derived from some other mild conditions on test spaces, as discussed in Appendix~\ref{monoentropic}}, and (ii) $\Omega$ is finite dimensional and
closed.

As is easily checked, local finiteness guarantees that the maximal
state space $\Omega({\frak A})$ is compact; thus, the closedness of
the physical statespace $\Omega$ insures that it, too, is
compact.\footnote{By \cite{Shultz}, any compact convex set can be
  represented as the full state space $\Omega({\mathfrak A})$ of some
  locally finite test space $\mathfrak A$.} It follows that every
state can be represented as a finite convex combination, or mixture,
of {\em pure states}, that is, extreme points of $\Omega$.  %

We now consider several examples. For us, a {\em classical system}
corresponds to a pair $(\{E\},\Delta(E))$ where the test space $\{E\}$
consists of a single measurement and where $\Delta(E)$ denotes the
entire simplex of probability weights on $E$. In other words, there is
just one test and any probability distribution over the outcomes is a
possible state.  A quantum system corresponds to $({\mathfrak
  F}(\H),\Omega(\H))$, where ${\mathfrak F}(\H)$ is the set of
(unordered) orthonormal bases on a complex Hilbert space $\H$ and
$\Omega(\H)$ is the set of density operators.\footnote{To be a bit
  more precise: a quantum state is the quadratic form associated with
  a density operator. We shall routinely identify a density operator
  $\rho$ with its quadratic form, writing $\rho(x)$ for $\langle \rho
  x, x \rangle$ where $x$ is a unit vector on $\H$.}

A simple example that is neither classical nor quantum, and to which we shall
refer often, is the ``two-bit" test space ${\mathfrak A}_2 =
\{\{a,a'\},\{b,b'\}\}$, consisting of a pair of two-outcome tests,
depicted in Fig.~\ref{fig:squittestspace}. The full state space
$\Omega({\mathfrak A}_2)$ is isomorphic to the unit square $[0,1]^2$
under the map $\alpha \mapsto (\alpha(a), \alpha(b))$ and is depicted
if Fig.~\ref{fig:squit}. Accordingly, we shall call a system of this
form a {\em square bit} or {\em squit}. A PR box is a particular
entangled state of two squits, as discussed below in Section
\ref{subsec:vanDam}.
\begin{figure}[ht]
\includegraphics[scale=0.6]{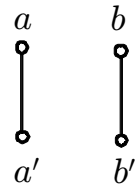}
\caption{The ``two-bit'' test space ${\mathfrak A}_2 = \{\{a,a'\},\{b,b'\}\}$.   It is depicted using a \emph{Greechie diagram}, wherein vertices denote outcomes and every smooth line through a set of vertices represents a test.}
\label{fig:squittestspace}
\end{figure}
\begin{figure}[H]
\centering
\includegraphics[scale=0.6]{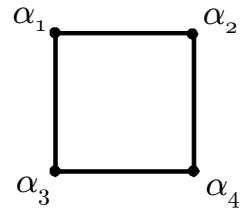}
\caption{The squit state space $\Omega({\mathfrak A}_2)$.  The pure state $\alpha_1$ yields the outcome $a$ in test $\{a,a'\}$ and the outcome $b$ in $\{b,b'\}$, that is, $\alpha_1(a)=1$, $\alpha_1(a') =0$, and $\alpha_1(b)=1$, $\alpha_1(b') =0$.  Similarly, $\alpha_2$, $\alpha_3$ and $\alpha_4$ yield the pairs of outcomes $a,b'$, $a',b$ and $a',b'$ respectively.}
\label{fig:squit}
\end{figure}

\section{Measurement and Mixing Entropies}\label{measurementandmixing}

Let $\H$ be a finite-dimensional Hilbert space, representing a quantum system. The von Neumann entropy of a state $\rho$ on this system is defined as $-\Tr(\rho \log\rho)$, where here and elsewhere, logarithms have base $2$. Equivalently, it is the Shannon entropy of the coefficients $\lambda_i$ in the spectral decomposition $\rho = \sum_{i} \lambda_i P_i$ (where the $P_i$ are $\rho$'s rank-one eigenprojections). In effect, the spectral decomposition privileges a particular
convex decomposition of the state, and (up to phases) a privileged test in ${\mathfrak F}(\H)$. In our much more general setting, where we have nothing like a spectral theorem, how might we define the entropy of a state? The following definitions suggest themselves.

\begin{definition} Let $\alpha$ be a state on ${\mathfrak A}$. For each test $E \in
{\mathfrak A}$, define the {\em local measurement entropy} of $\alpha$
at $E$, $H_{E}(\alpha)$, to be the classical (Shannon) entropy of
$\alpha|_{E}$, i.e.,
\[H_{E}(\alpha) := - \sum_{x \in E} \alpha(x) \log(\alpha(x)).\]
The {\em measurement entropy} of $\alpha$, $H(\alpha)$, is the
infimum of $H_{E}(\alpha)$ as $E$ ranges over ${\mathfrak
A}$,
i.e.,
\[H(\alpha) := \inf_{E \in {\mathfrak
A}} H_{E}(\alpha).\]
\end{definition}
Note that the measurement entropy of a state of $A = ({\mathfrak A}, \Omega)$ depends entirely on the structure of ${\mathfrak A}$, and is independent of the choice of state space $\Omega$. It will often be convenient to write $H(\alpha)$ as $H(A)$, where context makes clear which state is being considered.

For the remainder of this paper we make, and shall make free use of, the assumption that the measurement entropy of a state is actually achieved on some test, i.e., that $H(\alpha) = H_{E}(\alpha)$ for some $E \in {\mathfrak A}$. This is the case in quantum theory, and can be shown to hold much more generally, given some rather weak analytic requirements on an abstract model $({\mathfrak A},\Omega)$ -- for details, see Appendix~\ref{monoentropic}. It follows that $H(\alpha)=0$ if and only if there is a test such that $\alpha$ assigns probability $1$ to one of its outcomes.

\begin{definition} Let $\alpha$ be a state on ${\mathfrak A}$. The {\em
mixing} (or {\em preparation}) entropy for $\alpha$, denoted
$S(\alpha)$, is the infimum of the classical (Shannon) entropy
$H(p_1,...,p_n)$ over all finite convex decompositions $\alpha =
\sum_i p_i \alpha_i$ with $\alpha_i$ pure.
\end{definition}
Again, we write $S(A)$ for $S(\alpha)$ where $\alpha$ belongs to the state space $\Omega$ of a system $A = ({\mathfrak A},\Omega)$. In contrast to measurement entropy, the mixing entropy of a state depends only on the geometry of the state space $\Omega$, and is independent of the choice of test space ${\mathfrak A}$. The mixing entropy of a pure state is $0$.

Trivially, in classical probability theory, measurement and mixing
entropies coincide, both being simply the Shannon entropy. Much less
trivially, measurement and mixing entropies also coincide in quantum
theory, where they equal the von Neumann entropy.\footnote{A proof that the von Neumann entropy minimizes mixing entropy can be found in
\cite{BengtssonZyckowski}. The key observation is that the
mixing coefficients for any ensemble for $\rho$ can be obtained from
the eigenvalues of $\rho$ by a doubly stochastic matrix
(Shrodinger's Lemma), which can only increase entropy. An easier
version of the same argument (also noted by Hein \cite{Hein}) shows that the spectral decomposition
also minimizes measurement entropy.}  As the following example shows, however, measurement and mixing entropies can be quite different. 
\begin{example} (The firefly model \footnote{This example, well-known in the quantum-logical literature, has a fairly concrete interpretation in terms of a firefly in a three-chambered triangular box. See \cite{Wilce09} for details.}) Let $\mathfrak A = \{\{a,x,b\},\{b,y,c\},\{c,z,a\}\}$. This test space is depicted in Fig.~\ref{fig:firefly}. One can
check that $\Omega({\mathfrak A})$ has five pure states, one of which is
given by $\alpha(a) = \alpha(b) = \alpha(c) = 1/2, \alpha(x) =
\alpha(y) = \alpha(z) = 0$: since this is pure, $S(\alpha) = 0$, yet
$H(\alpha) = 1$. On the other hand, consider the pure states $\beta$
and $\gamma$ determined by $\beta(b) = \beta(z) = 1$ and $\gamma(x)
= \gamma(y) = \gamma(z) = 1$: their average, $\omega := 1/2\beta +
1/2\gamma$ has mixing entropy $S(\omega) = 1$. This follows from the fact that the only convex decomposition of $\omega$ into pure states is into $\beta$ and $\gamma$, which in turn follows from the fact that these are the only pure states that assign probability one to $z$. On the other hand, $\omega(z) = 1$,
so $H(\omega) = 0$.
\end{example}
\begin{figure}[H]
\centering
\includegraphics[scale=0.6]{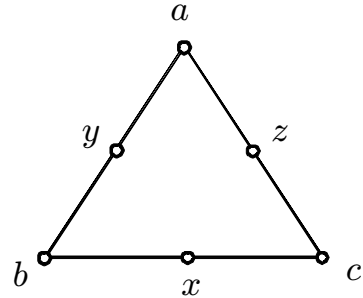}
\caption{The Greechie diagram for the test space of the firefly model.}
\label{fig:firefly}
\end{figure}

Even in the general case, measurement entropy is quite well behaved. For example, it is easy to see that $H(\alpha)$ is continous as a function of $\alpha$. Further,
\begin{theorem}
Measurement entropy is concave, i.e., if $\sum_i t_i \alpha_i$ is a convex combination of states $\alpha_i$ on $\mathcal{A}$, then
\begin{equation}
H\left( \sum_i t_i \alpha_i \right) \geq \sum_i t_i H \left( \alpha_i \right) .
\end{equation}
\end{theorem}
\begin{proof}
Since for each test $E$ the local entropy $H_{E}$ is concave,
\begin{eqnarray*}
 H\left(\sum_i t_i \alpha_i\right)  & = &  \inf_{E} H_{E}\left(\sum_i t_i \alpha_i\right) \\
 & \geq &  \inf \sum_i t_i H_{E}(\alpha_i) \geq \sum_i t_i
H(\alpha_i).
\end{eqnarray*}
\end{proof}

Mixing entropy is, by contrast, a curious beast. The following example shows that it need not be continuous as a function of the state.
\begin{example}\label{noncontinuousexample}
Let $\Omega \subseteq {\mathbb R}^{3}$ be the convex hull of the circle $C = \{(x,y,0) | x^2 + y^2 = 1\}$ and the line
segment $I = \{ (1,0,t) | -1 \leq t \leq 1 \}$ (Figure~\ref{noncontinuousfigure}). Let $\alpha$ denote the point of intersection of $I$ and $C$, i.e., the point $(1,0,0)$. The extreme points of this set are evidently the endpoints of $I$, together with the
points of $C \setminus \{ \alpha \}$. Note that $\alpha$ has a unique decomposition as a mixture of
extreme points of $\Omega$, namely, as an equal mixture of the endpoints of $I$. Thus, $S( \alpha ) = 1$. On the other hand, $\alpha$ can be approached as closely as we like by extreme points belonging to $C \setminus \{ \alpha \}$, which have mixing entropy $0$. The mixing entropy is therefore discontinuous at $\alpha$.
\end{example}
\begin{figure}[H]
\centering
\includegraphics[scale=0.6]{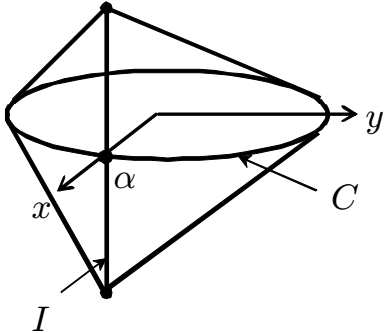}
\caption{Example of a state space $\Omega$ for which mixing entropy is not everywhere continuous (see Example~\ref{noncontinuousexample}).}
\label{noncontinuousfigure}
\end{figure}

\begin{example}{\em
Let $\Omega$ be a square. Let $\alpha$ and $\beta$ be the
midpoints of adjacent faces, noting that these each have unit mixing
entropy, $S(\alpha)=S(\beta)=1$. Let $\gamma = 1/2(\alpha + \beta)$ be the
mid-point of the line segment between $\alpha$ and $\beta$, and
note that it also lies on the line segment between antipodal vertices of $\Omega$ (the diagonal through the square between the chosen faces).  But given that $\gamma$ is not at the midpoint of this diagonal, the Shannon entropy for the associated convex decomposition is less than one, as is therefore the infimum over convex decompositions.  Therefore, the mixing entropy for $\gamma$ satisfies $S(\gamma)< 1$.  Consequently, $S(\gamma)< 1/2 S(\alpha) + 1/2 S(\beta)$, and we have a failure of concavity of the mixing entropy.}
\end{example}
\begin{figure}[H]
\centering
\includegraphics[scale=0.6]{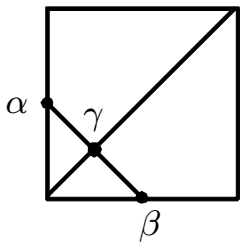}
\caption{Failure of concavity of mixing entropy
for a squit.}
\label{FIG:squareFC}
\end{figure}

In fact, the failure of concavity for the mixing entropy is quite generic.
\begin{theorem} \label{mixnonconcavetheorem}
Mixing entropy is not concave whenever the state space $\Omega$ is a non-simplicial polytope.
\end{theorem}
The proof is given in Appendix~\ref{mixnonconcave}. It follows that an
assumption of concavity for the mixing entropy forces the state space
to be either a simplex (i.e. classical) or not a polytope.  Hence such
an assumption or one that implies it may be a useful tool in
axiomatizing quantum theory.

It is natural to ask what follows from the condition that, as in
classical and quantum theories, measurement and mixing entropy
coincide. One immediate consequence is that mixing entropy will be
concave. In view of Theorem~\ref{mixnonconcavetheorem}, this implies
that either the system is essentially classical, or there are an
infinite number of pure states. Hence equality of measurement and
mixing entropies narrows down possible theories quite a lot. We
discuss this matter further in Appendix~\ref{monoentropic}.

Both measurement and mixing entropy have been considered before,
notably by Hein \cite{Hein}, in a similar context, albeit with
somewhat different aims than ours in view. There are various other
entropic quantities one could reasonably consider. \magenta For
example, a concept of entropy that might be more closely related to
operational tasks is the supremum, over convex decompositions of the
state and over tests, of the classical mutual information between the
random variable specifying the element of the convex decomposition,
and the random outcome of the test.  Natural analogues of this
quantity and of the measurement and preparation entropies defined
above exist in the closely related ordered linear spaces framework
(also known as the convex sets framework) for theories.  Test space
models such as we have defined above induce ordered linear spaces
models by a linearization procedure that embeds the test space in a
vector space and identifies outcomes in the test space with certain
elements of the dual vector space; this procedure allows one to define
concepts of measurement entropy more tightly related to the geometry
of the state space, but that can usually be viewed as special cases of
the test space definition.  Appendix \ref{appendix:linearization}
gives a further brief discussion of this.  \black

From this point on, we focus mainly on measurement entropy.  As always
with mathematical definitions, there is a certain tension between the
ideals of flexibility and generality, on the one hand, and, on the
other, the desire to avoid annoying pathologies. Our test-space
dependent definition of measurement entropy definitely errs on the
side of the former, in that it is consistent with quite absurd
examples. For example, if one includes in one's test space a test
having a single outcome, then all states will automatically have zero
entropy. One can avoid such difficulties by placing various
restrictions on the test spaces to be considered, at the cost of a
slightly more involved technical development. \magenta  Going to the linearized
setting mentioned above may also help. \black Our work in this paper
does not demand such fastidiousness, however, as our results are of a
very general character.

\section{Composite Systems and Joint Entropy}

Most of the interesting problems of information theory involve more
than one system. The following subsection describes how to treat
composite systems in the language of test spaces. The idea is that,
given systems $A$ and $B$, the joint system $AB$ should be associated
with a test space and state space of its own. However, there is
not a unique recipe for determining test and state spaces for $AB$
given the test and state spaces for $A$ and $B$.  Instead, a
theory must give additional rules that specify how systems
combine.\footnote{Quantum theory does just that: the rule is that the
  joint measurements and states correspond respectively to
  maximal sets of pairwise orthogonal projections and density
  operators on the tensor product of the individual Hilbert
  spaces.} Our results will pertain to a variety of notions of
composition, although we limit the scope by requiring certain
properties to hold.  In particular, we assume that the test space of
the composite includes all product tests and conditional two-stage
tests (where one party's choice of test is conditioned on the outcome
of the other party's test).  One motivation for this is to have a test
space that is sufficiently rich to be interesting.  Another motivation
is that this assumption guarantees that all states are
non-signaling.  We go on to define analogues of familiar
quantities, such as joint entropies and the mutual information, which
are used later to analyze information causality.

\subsection{Composite Systems}

Consider two systems, $A$ and $B$, where $A = ({\mathfrak A},\Omega^{A})$ and $B = ({\mathfrak B},\Omega^{B})$. For convenience, assume that these are controlled by two parties, called Alice and Bob. The first, and most basic, assumption we shall make is that Alice can perform any test $E \in {\mathfrak A}$ simultaneously with Bob performing any test $F \in {\mathfrak B}$. This can be regarded as a single {\em product test}. The possible outcomes of this product test are pairs of the form $(e,f)\in E \times F$.
\begin{definition}
The {\em Cartesian product} of the test spaces $\mathfrak A$ and $\mathfrak B$ is the collection of all product tests. It is denoted ${\mathfrak A} \times {\mathfrak B}$.
\end{definition}

The set $\Omega({\mathfrak A} \times {\mathfrak B})$, of all states that can be defined on the Cartesian product test space, typically includes \emph{signaling} states, which allow Alice to send messages instantaneously to Bob, or vice versa, by varying her choice of which test to perform.
\begin{definition}
A state $\omega^{AB}$ on ${\mathfrak A} \times {\mathfrak B}$ is {\em non-signaling} iff
\begin{eqnarray}\label{nosignaling}
\sum_{f\in F} \omega^{AB}(e,f) &=& \sum_{f'\in F'} \omega^{AB}(e,f') \quad \forall e,F,F' \nonumber\\
\sum_{e\in E} \omega^{AB}(e,f) &=& \sum_{e'\in E'} \omega^{AB}(e',f) \quad \forall f,E,E'.
\end{eqnarray}
\end{definition}
If a state $\omega^{AB}$ is non-signaling, it is possible to define the \emph{marginal} (or \emph{reduced}) state $\omega^A$ via
\begin{equation}
\omega^{A}(e) = \sum_{f \in F} \omega(e,f),
\end{equation}
where Eq.~(\ref{nosignaling}) ensures that the right hand side is independent of $F \in {\mathfrak B}$. The marginal $\omega^B$ is defined similarly.

If $\omega^{AB}$ is non-signaling, it is also possible to define a \emph{conditional state}, $\omega^{B|e}$.  Informally, this is the updated state at Bob's end following the outcome $e$ being obtained for a test at Alice's end:
\[
\omega^{B|e}(f) := \omega^{AB}(e,f)/\omega^A(e).
\]
By convention, $\omega^{B|e}$ is zero if $\omega^A(e)$ is zero. The conditional state $\omega^{A|f}$ is defined similarly.

Notice that a particular type of measurement, which might be thought reasonable, is not included in the Cartesian product. This is a joint measurement, where Alice first measures her system, and communicates the result to Bob, who performs a measurement which depends on Alice's outcome. Entangled measurements, such as are allowed in quantum theory, are also not included. Hence the Cartesian product ${\mathfrak A} \times {\mathfrak B}$ models a situation in which Alice and Bob are fairly limited - they can act independently and collate the results of their actions at a later time, but cannot otherwise communicate.

It is possible to construct a more sophisticated product of two test spaces, which does allow for the kind of two stage measurements just described (although still not entangled measurements). Let $\FR{\A}{\B}$ denote the test space consisting of the following
\begin{enumerate}
\item All two-stage tests, where a test $E \in {\mathfrak A}$ is performed, and then, depending on the
outcome $e$ that is obtained, a pre-selected test $F_{e} \in {\mathfrak B}$ is performed.
\item All two-stage tests, where a test $F \in {\mathfrak B}$ is performed, and then, depending on the
outcome $f$ that is obtained, a pre-selected test $E_{f} \in {\mathfrak A}$ is performed.
\end{enumerate}
$\FR{\A}{\B}$ is called the \emph{Foulis-Randall} or \emph{bilateral} product of test spaces ${\mathfrak A}$ and ${\mathfrak B}$.

The Foulis-Randall product contains the Cartesian product, ${\mathfrak
  A}\times{\mathfrak B} \subseteq \FR{\A}{\B}$, because product tests
are a special case of two-stage tests. Furthermore, if either $A$ or
$B$ is non-classical, then not all two-stage tests are product tests,
so that the containment is strict.  The containment of one test space
in another has consequences for their state spaces.  Specifically, if
$\mathfrak X$ and $\mathfrak Y$ are test spaces such that ${\mathfrak
  X} \subseteq {\mathfrak Y}$, then the convex set $\Omega({\mathfrak
  Y})$ may be in a higher dimensional space than $\Omega({\mathfrak
  X})$, but the restrictions of states in $\Omega({\mathfrak Y})$ to
${\mathfrak X}$ (which are well-defined, since every test in
${\mathfrak X}$ is also a test in ${\mathfrak Y}$) are all contained
in $\Omega({\mathfrak X})$. In other words, writing $\Omega({\mathfrak
  Y})|_{{\mathfrak X}}$ for the set of restrictions to $\mathfrak X$
of states on $\mathfrak Y$, we have $\Omega({\mathfrak
  Y})|_{{\mathfrak X}} \subseteq \Omega({\mathfrak X})$.  Because the
additional measurements in ${\mathfrak Y}$ place additional
constraints on these states, the containment may well be strict.

It follows that the restriction of the maximal state space of the
Foulis-Randall product to the Cartesian product is contained within
the maximal state space of the Cartesian product,
$\Omega(\FR{\A}{\B})|_{{\mathfrak A}\times{\mathfrak B}} \subseteq
\Omega({\mathfrak A}\times{\mathfrak B})$.  The containment is strict
if one of the systems is non-classical.  Indeed, the states in
$\Omega(\FR{\A}{\B})|_{{\mathfrak A}\times{\mathfrak B}}$ correspond
exactly to the non-signaling states in $\Omega({\mathfrak A} \times
{\mathfrak B})$.  This is demonstrated in Ref.~\cite{FR}.

We are now prepared to define the class of test and state spaces for
composites in which we shall be interested. The test space for the
composite, which we denote by ${\mathfrak C}$, is required to contain
the Foulis-Randall product of the components, $\FR{{\mathfrak
    A}}{{\mathfrak B}} \subseteq {\mathfrak C}$. The state space of
the composite, which we denote by $\Omega^{AB}$, is unconstrained
beyond being a subset of the maximal state space, $\Omega^{AB}
\subseteq \Omega({\mathfrak C})$.  Recalling that $\FR{{\mathfrak
    A}}{{\mathfrak B}} \subseteq {\mathfrak C}$ implies
$\Omega({\mathfrak C})|_{\FR{{\mathfrak A}}{{\mathfrak B}}} \subseteq
\Omega(\FR{{\mathfrak A}}{{\mathfrak B}})$ and that all the states in
$\Omega(\FR{{\mathfrak A}}{{\mathfrak B}})$ are non-signalling, it
follows that all states in $\Omega^{AB}$ are non-signalling.  Indeed,
the main motivation for confining our attention to test spaces
containing $\FR{{\mathfrak A}}{{\mathfrak B}}$ is that this is
sufficient to ensure no-signalling without any further constraints on
the state space.

Given a state $\omega^{AB}\in \Omega^{AB}$, the marginals $\omega^A$,
$\omega^B$, and conditionals of the form $\omega^{A|f}$,
$\omega^{B|e}$, are defined in the obvious way by the probabilities
which $\omega^{AB}$ assigns to the product tests.  Furthermore, we
assume that the composite systems we consider satisfy the following
natural requirement: that if a test is performed on system $A$, the
conditional state on system $B$ must be allowed in the theory, i.e.,
be contained in $\Omega^B$, and vice versa.  Hence $\Omega^{AB}$
satisfies the constraint that for all $e$ and $f$ such that
$\omega^A(e), \omega^B(f) \neq 0$, $\omega^{B|e}$ and $\omega^{A|f}$
belong to $\Omega^{B}$ and $\Omega^{A}$ respectively. This is enough
to ensure that the marginal states $\omega^{A}$, $\omega^B$ also
belong to the state spaces of the component systems.

A general composite test space ${\mathfrak C}$ may contain non-product measurements, which are not contained in the Foulis-Randall product. Quantum theory, for instance, has a test space for composites that is larger than the Foulis-Randall product.  If $A$ and $B$ are quantum systems, so that $A = ({\mathfrak F}(\H),\Omega(\H))$ and $B  = ({\mathfrak F}(\K),\Omega(\K))$, then the quantum joint system is $AB := ({\mathfrak F}(\H \otimes \K), \Omega(\H \otimes \K))$, which is a composite in our sense and contains non-product measurement outcomes, for instance, entangled ones.

Henceforth, $AB$ will stand for a general non-signaling
composite \blue of systems $A$ and $B$. In \black the particular case where $A =
(\{E\},\Delta(E))$ is a classical system, we always take \blue ${\mathfrak
  C}$ to be the Foulis-Randall product $\FR{\{E\}}{{\mathfrak B}}$ \black. We also assume that
composition of systems is associative, so that for any three systems
$A$, $B$ and $C$, there is a natural isomorphism $A(BC) \simeq (AB)C$.

In addition to specifying how systems combine, a probabilistic
theory must specify what sorts of systems are allowed.  For instance,
in finite-dimensional quantum theory, every dimensionality of Hilbert
space defines a different type of system and they are all allowed.
Furthermore, a classical system of arbitrary dimensionality (that is,
arbitrary cardinality for the test) can be defined within quantum
theory as a restriction upon a quantum system of the same
dimensionality, so in this sense classical systems are allowed as
well.  A probabilistic theory must specify the types of systems that
are allowed and how these compose.  We shall confine our attention to
theories incorporating only finite-dimensional systems, and those that
contain, for any finite set $E$, the classical system
$(E,\Delta(E))$. (Thus, for us, {\em quantum theory} means
finite-dimensional quantum theory in conjunction with classical
systems.) For a discussion of what such theories might look like in
category-theoretic terms, see \cite{BW09a, BW09b}.

\subsection{Joint Entropies, Conditional Entropies, Mutual Information}

Consider a composite system $AB=( {\mathfrak C}, \Omega^{AB})$. The measurement entropy  $H(\omega^{AB})$ of a state $\omega^{AB}\in\Omega^{AB}$, which we will sometimes denote by $H(AB)$, is the infimum over $E\in{\mathfrak C}$ of $H_E(\omega^{AB})$.
In this context, it will also be understood that $H(A)$ and $H(B)$ stand for
the entropies $H(\omega^{A})$ and $H(\omega^{B})$ of the marginal
states $\omega^{A}$ and $\omega^{B}$.
\begin{theorem}\label{measuremententsubadditive}
Measurement entropy is subadditive. That is, for any composite $AB$,
\[H(AB) \leq H(A) + H(B).\]
\end{theorem}
\begin{proof} Let $\omega$ be the joint state of $AB$, with marginal states $\omega^A$ and $\omega^B$. Choose $E$ and $F$ with $H(\omega^A) = H_{E}(\omega^A)$ and $H(\omega^{B}) = H_{F}(\omega^{B})$. By the definition of
measurement entropy, the definition of a composite, and the subadditivity of Shannon entropy, we have
$H(AB) \leq H_{EF}(\omega) \leq H_{E}(\omega^{A}) + H_{F}(\omega^{B})  = H(A) + H(B)$.
\end{proof}

\begin{definition}
The \emph{conditional measurement entropy} between $A$ and $B$ is defined to be
\begin{equation}\label{condmmtentropy}
H(A|B) := H(AB) - H(B).
\end{equation}
\end{definition}
Our notation here is less precise than it might be, since the joint entropy $H(AB)$ depends on the test space associated with the joint system, hence so do conditional entropies. We will try to be clear, at any point where the question could
arise, as to what product is in play.

Classically, given a joint distribution $\omega^{AB}$ over variables $A$ and $B$, one defines the {\em mutual information} by
\begin{equation}\label{mutualinfdef}
I(A:B) = H(A) + H(B) - H(AB),
\end{equation}
where $H$ denotes the Shannon entropy. One can regard this as a
measure of how far $A$ and $B$ are from being independent: by
subadditivity, $I(A:B) \geq 0$, with $I(A:B) = 0$ iff $A$ and $B$ are
independent, i.e., $\omega^{AB}$ factorizes. In attempting to extend
the concept of mutual information to more general models, one might
very naturally consider defining $I(A:B)$ to be the maximum of the
mutual informations $I(E:F)$ as $E$ and $F$ range over tests belonging
to systems $A$ and $B$, respectively. However, the usual practice in
quantum theory is simply to take Equation~(\ref{mutualinfdef}), with
von Neumann entropies replacing Shannon entropies, as {\em defining}
mutual information. In general, this gives a different value. In order
to facilitate comparison with quantum theory, we shall adopt the
following
\begin{definition}
Let $AB$ be a composite system. The {\em
  measurement-entropy-based mutual information} between $A$ and
$B$ is
\begin{equation} \label{mutualinfdefgeneral}
I(A:B) := H(A) + H(B) - H(AB).
\end{equation}
\end{definition}
With this definition, the subadditivity of measurement entropy
(Theorem~\ref{measuremententsubadditive}) implies that
measurement-entropy-based mutual information is
non-negative. Hereafter, we will refer to this simply as the
``mutual information''.  Note that Eq.~(\ref{mutualinfdef}) is a
special case of this definition.

Now intuitively, one might expect that the mutual information $I(A:B)$
between two systems should not {\em decrease} if we recognize that $B$
is a part of some larger composite system $BC$ -- i.e., that $I(A:B)
\leq I(A:BC)$. Simple algebraic manipulations (using
Eqs.~(\ref{condmmtentropy}) and (\ref{mutualinfdefgeneral}))
allow us to reformulate this condition in various ways.
\begin{lemma}\label{strongsublemma}
The following are equivalent:
\begin{itemize}
\item[(a)] $I(A:BC) \geq I(A:B)$
\item[(b)] $H(A|BC) \leq H(A|B)$
\item[(c)] $H(A,B) + H(B,C) - H(B) \leq H(A,B,C)$
\item[(d)] $I(A:B|C) \geq 0$, where $I(A:B|C) = H(A|C) + H(B|C) - H(AB|C)$.
\end{itemize}
\end{lemma}
\begin{definition}\label{strongsubadd}
The measurement entropy is said to be \emph{strongly subadditive} if it satisfies the equivalent conditions (a)-(d).
\end{definition}
(We use this terminology despite the fact that it is usually only condition (c) that goes by the name of ``strong subadditivity'' and despite the fact that conditions (a) and (d) constrain the measurement entropy only through the definitions of $I(A:BC)$ and $I(A:B|C)$.)
A probabilistic theory in which conditions (a)-(d) are satisfied for all systems $A, B$ and $C$ will also be called \emph{strongly subadditive}.

Both the Shannon and von Neumann entropies are strongly
subadditive. In the former case, this is a straightforward exercise,
but in the latter, a relatively deep fact. Colloquially, this means
that in classical and quantum theories, just forgetting about or
discarding a system $C$ never increases one's mutual information
between systems $A$ and $B$. As the following shows, however, strong
subadditivity can fail in general theories, even when two of the three
systems are classical. One potential gloss is that discarding or
forgetting about system $C$ \emph{can} increase the mutual information
between systems $A$ and $B$. But a more sensible
reading is perhaps that the quantity defined as mutual information
should not in the general case be interpreted as ``the information one
system contains about another.''

\begin{figure}[H]
\centering
\includegraphics[scale=0.6]{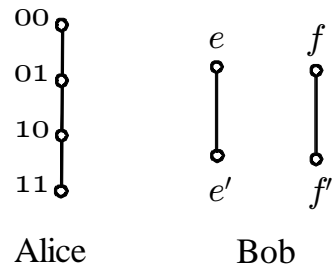}
\caption{The component test spaces for example \ref{failurestrongsubaddexample}.}
\label{FIG:2bitsandsquit}
\end{figure}
\begin{example}\label{failurestrongsubaddexample}(Failure of strong subadditivity of the measurement entropy.)
Consider a tripartite system $ABC$, where $A$ and $B$ are classical bits and $C$ is a squit, with test space $\{\{e,e'\},\{f,f'\}\}$. Consider the joint state described by the following table:
\[
\begin{array}{l|cc|cc}
     &    e    &    e'     &    f    &   f'  \\
     \hline
00  & 1/4   &  0    &  1/4  &  0     \\
01 &  1/4   &  0    &   0  &  1/4  \\
10 &   0    &  1/4  &   1/4  &   0  \\
11 &   0    &  1/4  &   0  &   1/4
\end{array}  \]
In words, the outcome of test $\{e,e'\}$ is perfectly correlated with $A$ while the outcome of the test $\{f,f'\}$ is perfectly correlated with $B$. It is easily verified that
\[
H(C) = H(AC) = H(BC) = 1.
\]
If all three systems are measured, with the test on $C$ perhaps depending on the values of $A$ and $B$, there are always four distinct outcomes, each with probability $1/4$. Hence $H(ABC) = 2$ and
\begin{eqnarray*}
I(A:B|C) & = & H(AC) + H(BC) - H(C) - H(ABC)\\
& = & 1 + 1 - 1 - 2 = -1 < 0,
\end{eqnarray*}
which contradicts form (d) of strong subadditivity.
\end{example}
Note that the foregoing example is all but classical, depending not on any notion of entanglement or non-locality, but only on
the fact that one can measure either, but never both, of $\{e,e'\}$ and $\{f,f'\}$.

This section concludes with some lemmas, which hold in the special case that one or more of the systems in the composite is classical. Some of these are useful later on.

\begin{lemma}\label{goodgollymissmolly}
Let $\omega^{AB}$ be a state on $AB$, where $A$ is classical. Then
\begin{equation} \label{molly2}
H(\omega^{AB}) = H(\omega^A) + \sum_{e \in E} \omega^A(e) H(\omega^{B|e}).
\end{equation}
\end{lemma}
The proof is straightforward. As a shorthand, when $A$ is classical we might write
\[
H(AB) =  H(A) + \sum_e p(e) H(B|e).
\]
\begin{corollary}
If $A$ is classical, then $H(B|A) \geq 0$ for any system $B$.
\end{corollary}
The proof is immediate from Eqs.~(\ref{condmmtentropy}) and (\ref{molly2}).
\begin{corollary}\label{independentsystemscorollary}
If $A$ is classical and independent of $B$, then $H(AB) = H(A) + H(B)$.
\end{corollary}
\begin{proof}
The assertion that $A$ and $B$ are independent means that the joint state is $\omega^{AB} = \omega^A \otimes \omega^B$, i.e., that $\omega^{B|e} = \omega^B$ for all $e \in E$. By Lemma~\ref{goodgollymissmolly}, we have
\begin{eqnarray*}
H(AB) & = & H(A) + \sum_{e \in E} \omega^A(e) H(\omega^{B|e})\\
 & = &  H(A) + \left(\sum_{e \in E} \omega^A(e)\right) H(B) = H(A) + H(B).
\end{eqnarray*}
\end{proof}

Finally, strong subadditivity does hold in the special case that systems $A$ and $C$ in Lemma~\ref{strongsublemma} are classical. Colloquially, discarding a \emph{classical} system can never result in an increase in the mutual information between a general system and another classical system.
\begin{lemma}
Let $A$ and $C$ be classical. Then for any system $B$,
\[
H(A|BC) \leq H(A|B).
\]
Hence, the equivalent conditions of Lemma~\ref{strongsublemma} are satisfied.
\end{lemma}
\begin{proof}
Let $A = \{E\}$, $C = \{G\}$, and let the joint state
of $ABC$ be $\omega^{ABC}$. Then the marginal state of $BC$ satisfies
$\omega^{BC}(fg) = \omega^C(g) \cdot \omega^{B|g}(f)$ where, for all $g
\in G$
\[
\omega^{B|g} = \sum_{e \in E} \frac{\omega^{AC}(eg)}{\omega^{C}(g)} \omega^{B|eg}.
\]

By Lemma~\ref{goodgollymissmolly}, we have
\begin{eqnarray*}
H(A|BC) &=&  H(ABC) - H(BC)\\
 &=&  H(AB) + \sum_{g\in G}\sum_{e \in E} \omega^{AC}(eg) H(\omega^{B|eg}) \\ && - H(C) - \sum_{g\in G} \omega^C(g) H(\omega^{B|g}).
\end{eqnarray*}
We can rewrite this as
\begin{eqnarray}\label{johnnybgoode}
H(A|BC) &=& H(A|C) \nonumber\\
&& + \sum_{g\in G} \sum_{e\in E} \omega^{AC}(eg) H(\omega^{B|eg}) \nonumber \\
&& - \sum_y \omega^C(g)H(\omega^{C|g}).
\end{eqnarray}
Since measurement entropy is concave,
\[
H(\omega^{C|g}) \geq \sum_{e \in E} \frac{\omega^{AC}(eg)}{\omega^C(g)}H(\omega^{B|g}),
\]
whence
\[
\sum_{g} \omega^C(g) H(\omega^{B|g}) \geq \sum_{g \in G} \sum_{e \in E} \omega^{AC}(eg) \omega^{C|eg}.
\]
It follows that $\sum_{e,g} \omega^{AC}(eg) H(\omega^{B|eg}) - \sum_{y} \omega^C(g) H(\omega^{B|g}) \leq 0$, which, combined with Equation~(\ref{johnnybgoode}), gives the desired result that $H(A|BC) \leq H(A|C)$.
\end{proof}

\subsection{Data Processing and the Holevo Bound}
\label{holevo}

A fundamental result in quantum information theory, the {\em Holevo bound}, asserts that if Alice prepares a quantum state $\rho = \sum_{x \in E} p_x \rho_x$ for Bob, then, for any  measurement $F$ that Bob can make on his system,
\[
I(E:F) \leq \chi,
\]
where $\chi := H(\rho) - \sum_{x \in E} p_x H(\rho_x)$ (often called the {\em Holevo quantity}).

This inequality makes sense in our more general setting. Suppose that Alice has a classical system $A = (\{E\},\Delta(E))$ and Bob a general system $B$. Alice's system is to serve as a record of which state of $B$ she prepared. Hence the situation above is modeled by the joint state $\omega^{AB} = \sum_{x \in E} p_x  \delta_{x} \otimes \beta_x$, where $\delta_x$ is a deterministic state of Alice's system with $\delta_x(x) = 1$. Bob's marginal state is $\omega^{B} = \sum_{x \in E} p_x \beta_x$. By Lemma~\ref{goodgollymissmolly}, $H(\omega^{AB}) = H(A) + \sum_{x \in E} p_x H(\beta_x)$. Hence,
\begin{eqnarray*}
I(A:B) & = & H(A) + H(B) - H(AB) \\
& = & H(A) + H(B) - \left( H(A) + \sum_{x \in E} p_x H(\beta_x) \right) \\
& = & H(\omega^B) - \sum_{x \in E} p_x H(\beta_x) = \chi.
\end{eqnarray*}
Accordingly, the content of the Holevo bound is simply that the mutual information between the measurement of Alice's classical system and any measurement on Bob's system is no greater than $I(A:B)$,
\[
I(E:F) \leq I(A:B).
\]

While this is certainly natural, it does not always hold.

\begin{example} (Failure of the Holevo bound.)
Let $A$ be a classical bit, $A = \{\{0,1\}\}$ and $B$ a squit, $B = \{F = \{f,f'\}, G = \{g,g'\}\}$, and consider the state
\[
\begin{array}{l|cc|cc}
     &    f    &    f'     &    g    &   g'  \\
     \hline
0  & 1/2   &  0    &  1/2  &  0     \\
1 &  1/2   &  0    &   0  &  1/2
\end{array}
\]
It is easy to check that $H(A) = H(AB) = H_G(B) = 1$, and $H(B) = H_{F}(B) = 0$. Hence,
\[
I(A:B) = H(A) + H(B) - H(AB) = 1 + 0 - 1 = 0
\]
while
\[
I(E:F) = 1 + 1 - 1 = 1 > 0.
\]
\end{example}

Both strong subadditivity and the Holevo bound are instances of a more
basic principle.  The {\em data processing inequality} (DPI) asserts
that, for any systems $A$ and $B$ and any physical process ${\cal E} :
B \rightarrow C$,\footnote{Clearly, the content of the DPI depends on
  the definition of mutual information, which in turn depends on the
  definition of entropy used. As before, we continue to define mutual
  information in terms of the measurement entropy, as in
  Definition \ref{mutualinfdefgeneral}. But note that if the DPI fails for one
  definition of mutual information, the corresponding condition may
  hold for another - indeed this may be a reason to prefer the
  latter. In this paper, we omit a formal treatment of ``processes''.}
\[
I(A:{\cal E}(B)) \leq I(A:B).
\]
The strong subadditivity of entropy amounts to the DPI for the process that simply discards a system (the \emph{marginalization map} $BC \rightarrow C$). The Holevo bound is the DPI for the special case of measurements, which can be understood as processes taking a system into a classical system which records the outcome.

It seems reasonable that discarding a system, or performing a measurement, should be allowed processes in a physical theory. But a notion of mutual information, according to which discarding a system, or performing a measurement, causes a \emph{gain} of mutual information seems bizarre.  So it is an attractive idea that a physical theory should allow at least some definition of entropy and mutual information such that the corresponding DPI is satisfied.

\section{Information Causality}

In \cite{Petal}, Pawlowski {\em et al.} define a principle they call {\em Information Causality} in terms of the following protocol. Alice and Bob share a joint non-signaling state, known to both parties.
Alice receives a random bit string $E$ of length $N$, makes measurements, and
sends Bob a message $F$ of no more than $m$ bits. Bob receives
a random variable $G$ encoding a number $k = 1,...,N$, instructing him to guess the value of Alice's $k$th bit $E_{k}$. Bob thereupon makes a suitable measurement and, based upon its outcome, and the message from Alice, produces
his guess, $b_k$. Information causality is the condition that
\begin{equation}\label{ic}
\sum_{k=1}^{N} I(E_k : b_k | G = k) \leq m.
\end{equation}

The main result of \cite{Petal} is that if a theory contains states
that violate the CHSH inequality \cite{chsh} by more than the
Tsirel'son bound \cite{Tsirelson}, then it violates information
causality. In particular, if Alice and Bob can share PR boxes, then
using a protocol due to van Dam \cite{vanDam}, they can violate
information causality maximally, meaning that Bob's guess is correct
with certainty, and the left hand side of Equation~(\ref{ic}) is
$N$. Pawlowski \emph{et al.} also give a proof, using fairly standard
manipulations of quantum mutual information, that quantum theory {\em
  does} satisfy information causality.

Having seen how to define notions of entropy and mutual information
for general systems, it is interesting to consider where Pawlowski
\emph{et al.}'s quantum proof breaks down for some non-quantum systems
such as PR boxes. One issue is that the proof uses strong
subadditivity. As the following subsection shows, in the case where a
PR box is the shared state, the van Dam protocol itself provides an
example of the failure of strong subadditivity of the measurement
entropy. Section~\ref{dpiic} provides a converse result. Any theory
which is monoentropic, strongly subadditive and where the Holevo bound
holds, must satisfy information causality.

First, a few words about how to describe this setting in our
terminology. Let Alice and Bob share two systems $A$ and $B$, where
each of these, as usual, has an associated test space. The joint test
space of $AB$ is immaterial, as long as it includes the Foulis-Randall
product (i.e., allows all the separable measurements). The bit strings
$E$ and $F$ are regarded as classical systems in their own right and
the joint test space for a classical and a general system is, as
always, assumed to be the Foulis-Randall product. Systems $A$ and $B$
begin the protocol in some joint non-signaling state $\omega^{AB}$.

\subsection{The van Dam protocol.}\label{subsec:vanDam}

Consider a special case of the protocol described above, in which Alice and Bob share a PR box. Alice is supplied with a two-bit string $E=E_1E_2$, and transmits one bit $F$ to Bob. Let the PR box be a state of two systems $A$ and $B$, where $A$ and $B$ are squits corresponding to the test spaces $\{ \{a_1,a_1' \}, \{a_2,a_2'\} \}$ and $\{ \{b_1,b_1' \}, \{b_2,b_2'\} \}$ respectively. The joint state of $A$ and $B$ is
\[
\begin{array}{c|cc|cc}
& a_1  &  a_1'     &   a_2 & a_2' \\
\hline
b_1 &   &  1/2 &   & 1/2 \\
b_1'&   1/2  &  & 1/2  &  \\
\hline
b_2 &    & 1/2 &  1/2 & \\
b_2' & 1/2 &   &    & 1/2
\end{array}
\]
It can be verified that these outcome probabilities are indeed the PR box correlations, violating the CHSH inequality maximally. In van Dam's protocol, Alice determines the parity, $E_1\oplus E_2$ (where $\oplus$ denotes addition mod $2$). If this is zero she performs the $\{a_1,a_1'\}$ measurement on her system; if it is $1$, she performs the $\{a_2,a_2'\}$ measurement. She then sends Bob a single bit with a value equal to the parity of her outcome and $E_1$ (where unprimed outcomes correspond to $0$ and primed outcomes to $1$). Bob can then determine the value of $E_1$ by measuring $\{b_1,b_1'\}$, or the value of $E_2$ by measuring $\{b_2,b_2'\}$.

Consider now an intermediate stage in this protocol, at which Alice has measured her system, and sent the bit $F$ to Bob, who has not yet measured his system. Bob has access to systems $B$ and $F$, but does not know the outcome of Alice's measurement. Hence consider the joint state of $EFB$, averaged over the outcomes of Alice's measurement. This is easily verified to be
\[
\begin{array}{c|cc|cc}
E_1 E_2 F   &    b_1    &    b_1'     &    b_2    &   b_2'  \\
   \hline
000  & 1/8 &           & 1/8 &       \\
001 &         & 1/8   &         &  1/8 \\
111  & 1/8 &           & 1/8 &  \\
110 &         & 1/8   &         &  1/8 \\
010  & 1/8 &           &         &  1/8 \\
011 &         & 1/8   & 1/8 &  \\
101  & 1/8 &           &         & 1/8   \\
100 &         & 1/8   & 1/8 &
\end{array} \]
Minimizing over the possible measurement choices on $B$,
\[
H(E_1,F,B) = H(E_2,F, B) = H(F, B) = 2.
\]
But clearly, $H(E,B) = 3$, so
\[
I(E_1:E_2 | F,B) = 2 + 2 - 2 - 3 = -1 < 0.
\]

\subsection{Theories satisfying information causality}\label{dpiic}

As the previous subsection observes, the van Dam protocol involves a joint state on a classical-nonclassical composite system, which does not satisfy strong subadditivity of entropy. This is enough to prevent the proof of information causality going through. This subsection proves a converse result.
\begin{theorem} \label{ictheorem}
Suppose that a theory is
\begin{enumerate}
\item monoentropic, meaning that measurement entropy equals mixing entropy for all systems.
\item Strongly subadditive
\item Satisfies the Holevo bound.
\end{enumerate}
Then the theory satisfies information causality. It follows that any theory satisfying these conditions cannot violate Tsirel'son's bound.
\end{theorem}
Note that, as discussed in Section~\ref{holevo}, the second and third conditions both follow from a single assumption of a data processing inequality. Note also that in proving Theorem~\ref{ictheorem}, the condition that a theory be monoentropic is only used to establish the technical condition that $H(A|B) \geq 0$ when $A$ is classical.\footnote{This does not hold in all theories.} So the theorem would still be valid if the monoentropic assumption were replaced by a direct assumption that for classical $A$, $H(A|B) \geq 0$. Otherwise, begin with
\begin{lemma}\label{unientlemma}
Suppose that a theory is monoentropic and that $A$ is a classical system. Then $H(A|B) \geq 0$ for any system $B$.
\end{lemma}
\begin{proof}(Lemma~\ref{unientlemma}.)
Suppose that $A$ is a classical system, and that the joint state of $AB$ is $\omega^{AB}$. If the measurement and mixing entropies are equal, then Lemma~\ref{goodgollymissmolly} immediately gives
\[
S(AB) = S(A) + \sum_x p_x S(\beta_x),
\]
where $p_x = \omega^A(x)$ and $\beta_x$ is the state of $B$ conditioned on $x$. Recall that the mixing entropy of a state is defined in terms of an infimum over convex decompositions into pure states. For a fixed $\epsilon$, call a convex decomposition of a state $\omega$ into pure states $\epsilon$-\emph{optimal} if the Shannon entropy of the coefficients is $\leq S(\omega) + \epsilon$. For any $\epsilon>0$, there is an $\epsilon$-optimal decomposition. Let
\[
\beta_x = \sum_y q_{y|x} \beta_{xy}
\]
be an $\epsilon$-optimal convex decomposition of $\beta_x$ into pure states $\beta_{xy}$. It follows that
\[
\omega^B = \sum_x \sum_y p_x q_{y|x} \beta_{xy}
\]
is a (possibly far from optimal) convex decomposition of $\omega^B$ into pure states. Hence $S(B)$ is less than or equal to the Shannon entropy of the coefficients on the right hand side. Therefore
\begin{eqnarray*}
S(B) &\leq& H(p_x) + \sum_x p_x H(q_{y|x}) \\
&\le & S(A) + \sum_x p_x (S(\beta_x)+\epsilon) \\
&=& S(AB)+\epsilon.
\end{eqnarray*}
Since this holds for any $\epsilon$, we have $S(B)\leq S(AB)$ and $S(A|B) = S(AB) - S(B) \geq 0$ as required.
\end{proof}

Given Lemma~\ref{unientlemma}, the proof of Theorem~\ref{ictheorem} is essentially a reconstruction of the quantum argument of Appendix~A of \cite{Petal}, adapted to the broader setting of non-signaling states on test spaces. In its form the proof is the same, but great care must be taken at each step to ensure that the relevant properties of entropies and mutual information still hold. Many of the steps still go through in virtue of generic properties of the measurement entropy. The explicit assumptions of Theorem~\ref{ictheorem} are needed for the rest.

\begin{proof}(Theorem~\ref{ictheorem}.)
Assume that Alice and Bob share a joint system $AB$. Consider the $N$-bit string which Alice receives as a classical system $E$, and consider the $m$-bit message which Alice sends to Bob as a classical system $F$. Let $E_k$ denote Alice's $k$th bit. Consider the stage of the protocol where Alice has measured system $A$, and sent $F$ to Bob, but Bob has not yet measured system $B$. Bob has control of systems $F$ and $B$ at this point, and does not know the outcome of Alice's measurement. Hence the strategy is to consider the joint state of the systems $E$, $F$ and $B$, averaged over Alice's outcomes.

The first goal is to show that the joint state at this point satisfies
\begin{equation}\label{mutinflessm}
I(E:FB) \leq m.
\end{equation}
By the fact that the initial state of $AB$ is non-signaling, $E$ is independent of $B$. Therefore Corollary~\ref{independentsystemscorollary} yields $I(E:B) = 0$. Using this, along with the definitions and straightforward algebraic manipulation, we get
\begin{eqnarray*}
I(E:FB) &=& I(E:B) + I(E:F|B) \\
&=& I(E:F|B) \\
&=& I(EB:F) - I(B:F).
\end{eqnarray*}
By Theorem~\ref{measuremententsubadditive}, mutual information is non-negative, so
\begin{equation}
I(E:FB) \leq I(EB:F).
\end{equation}
Now, $I(EB:F) = H(EB) + H(F) - H(EFB) = H(F) - H(F|EB)$. By the assumption that the theory is monoentropic, and Lemma~\ref{unientlemma}, $H(F|EB) \geq 0$. So
\begin{equation}
I(EB:F) \leq H(F) \leq m.
\end{equation}
This gives Equation~(\ref{mutinflessm}).

The next step is to establish
\begin{equation}\label{mutinfsumlessmutinf}
\sum_ {k=1}^N   I(E_k:FB) \leq I(E:FB).
\end{equation}
Rearrangement of definitions yields
\begin{eqnarray*}
I(E:FB) &=& I(E_1\ldots E_N:FB)\\
&=& I(E_1:FB) + I(E_2\ldots E_N:FB|E_1)
\end{eqnarray*}
and
\begin{eqnarray*}
I(E_2\ldots E_N:FB|E_1) &=& \\
&& \hspace{-80pt} I(E_2\ldots E_N:FBE_1) - I(E_2\ldots E_N:E_1).
\end{eqnarray*}
Since the distribution on $E$ is uniform (the bits are independent),
$I(E_2\ldots E_N:E_1) = 0$. Hence,
\[
I(E:FB) = I(E_1 :FB) + I(E_2\ldots E_N:FBE_1).
\]
By strong subadditivity,
\[
I(E_2\ldots E_N:FBE_1) \geq I(E_2\ldots E_N:FB).
\]
So
\[
I(E:FB) \leq I(E_1 : FB) + I(E_2\ldots E_N:FB).
\]
Applying this inequality recursively gives
Equation~(\ref{mutinfsumlessmutinf}).

Finally, consider the last stage of the protocol. If Bob is instructed
to guess the $k$th bit, then, depending on the message $F$, he
measures system $B$. This can be seen as a single joint measurement
$X_k$ on the system $FB$. \footnote{Recall that the joint test space
  for a classical and general system is always assumed to be the
  Foulis-Randall product, and that this includes measurements of the
  form: measure $F$ first and, depending on the outcome, measure $B$.}
The Holevo bound, combined with
Equations~(\ref{mutinflessm},\ref{mutinfsumlessmutinf}) gives
\[
\sum_{k=1}^N I(E_k:X_k) \leq m.
\]
Finally, Bob outputs a guess $b_k$ for the value of $E_k$, where the
guess depends on $k$ and on the outcome of the measurement $X_k$. The
usual data processing inequality applied to classical mutual
information yields
\[
\sum_{k=1}^N I(E_k:b_k|G=k) \leq m,
\]
which is information causality.
\end{proof}
\vspace{10pt}

\section{Conclusions, Discussion and Further Questions}

We have defined preparation and measurement based generalizations of
quantum and classical entropy and mutual and conditional information,
and studied some of their basic properties.  We called theories in which
they coincide \emph{monoentropic}, and showed that if they in addition
satisfy the data processing inequality (or at least its corollaries
strong subadditivity and the generalized Holevo bound), Pawlowski
\emph{et al.}'s information causality principle holds.  By their remarkable
result that any correlations violating the Tsirel'son bound can be
used to violate information causality, it follows that monoentropic
theories satisfying data processing must, like quantum theory, obey
the Tsirel'son bound.  Monoentropicity, is a strong constraint on
theories, as we have shown by establishing that it fails for
all polytopes except simplices.

Our results indicate that it is interesting and profitable to
develop notions of entropy, and allied notions of conditional entropy
and mutual information, for abstract probabilistic models. This paper
should be regarded as only a preliminary exploration of this
possibility.

A natural direction for further research is to study data compression
and channel capacities in the abstract setting of this paper.  It is
natural to seek a measure of entropy that governs the rate of
high-fidelity data compression, as Shannon and von Neumann entropy do
in classical and quantum theory.  A first step toward exploring {\em
  classical} channel capacities in generalized probabilistic theories
might be to identify sufficient conditions for the Holevo bound to
hold.  This is related to the issue of finding an operationally
motivated definition of mutual information.  Arguably,
a properly motivated notion of mutual information should {\em
  manifestly} be monotonic.  Of course, the monotonicity of quantum
mutual information---equivalently, the strong subadditivity of quantum
entropy---is not manifest from its usual functional form.
Still, the outright failure of the measurement-entropy-based mutual
information to satisfy monotonicity in some cases raises a question as
to its significance.  Although in such cases measurement-based mutual
information cannot be used to establish information causality through
a proof parallel to Pawlowski \emph{et al.}'s quantum proof, it could
be that IC nevertheless holds in some such cases.  One should be
cautious, though, about dismissing natural generalizations of
classical quantities on the grounds that they fail to satisfy
intuitively compelling properties.  \red A case in point is the
history of skepticism, based on the fact that it can be negative,
about the operational significance of conditional information in
quantum information theory.  It was known for many years that the
conditional mutual information can be negative, but it was eventually
shown to have an operational interpretation, involving the rate for
quantum state merging protocols.  \black It is also good to keep in
mind that different operational motivations might turn out to be
naturally associated with different entropic quantities, each with
reasonable claim to be called mutual information.

At a more fundamental level, one would like to understand better the
operational significance of various notions of entropy for abstract
probabilistic models and theories.  It is likely that the entropic
quantities we have discussed here, measurement and mixing entropy,
will turn out not to be best notions of entropy to use in many
situations.  For example, in Appendix \ref{appendix:linearization}, we
considered a variation (or perhaps better, a specialization) of the
notion of measurement entropy that is more tightly coupled to the
geometry of the state space.

We have seen that, taken together, the conditions of monoentropicity,
strong subadditivity, and the Holevo bound, imply information
causality. It is not out of the question that some subset of these
conditions would suffice (especially since we need only very special
cases of strong subadditivity).  Alternatively, it would be of
interest to find a single, reasonably simple physical postulate that
would imply all three of these conditions. It seems plausible that
such a postulate exists.  On the one hand, strong subadditivity and the
Holevo bound are both special cases of the data processing inequality,
which in turn can be derived (as we will detail in a future paper)
from the assumption that arbitrary processes can be dilated to
reversible ones.  On the other hand, as we show in Appendix B,
monoentropicity can be derived from conditions of a similar flavor,
involving the dilatability of mixed states to pure states with 
a ``marginal steering'' property.  Another avenue to explore is the
consequence of monoentropicity that is needed for the IC proof:
positivity of conditional information when a classical system is
conditioned upon a general one.  Although its operational
interpretation is not evident at first blush, it warrants further
study.

We hope to discuss all of these matters in detail in a future paper.

\acknowledgments

This research was supported by the United States Government through
grant OUR-0754079 from the National Science Foundation.  It was also
supported by Perimeter Institute for Theoretical Physics.  Research at
Perimeter Institute is supported by the Government of Canada through
Industry Canada and by the Province of Ontario through the Ministry of
Research and Innovation.  This work was also supported by the EU's
FP6-FET Integrated Projects SCALA (CT-015714) and QAP (CT-015848), and
the UK EPSRC project QIP-IRC.  Jonathan Barrett is supported by an
EPSRC Career Acceleration Fellowship.  At IQC, Matthew Leifer was supported
in part by MITACS and ORDCF.  At Perimeter Institute, Matthew Leifer was
supported in part by grant RFP1-06-006 from The Foundational Questions Institute (fqxi.org).

\appendix

\section{Non-concavity of mixing entropy}\label{mixnonconcave}
In this appendix we prove Theorem \ref{mixnonconcavetheorem}, which
states that the mixing entropy is not concave for nonsimplicial
polytopes.  As a preliminary to the proof, we state some basic
definitions and facts that we will use.  A {\em face} of
a convex set $C$, which is a set $F \subseteq C$ such that every $x
\in C$ that can appear in a convex decomposition of something in $F$,
is also in $F$.  A {\em maximal face} of $C$ is one that is not a
proper subset of any face in $C$ other than $C$ itself.  An {\em
  exposed face} of $C$ is a subset of $C$ that is the intersection of
$C$ with a hyperplane supporting it (such a subset is easily shown to
be a face).  All faces of a polytope are exposed, and the maximal ones
have affine codimension $1$, i.e. their spans are affine hyperplanes.
We denote the affine space generated by a set $S$ by $\aff(S)$, the
linear span of $S$ by $\lin(S)$, and the cone generated by $S$
(i.e. the set of nonnegative linear combinations of elements of $S$)
by $\cone(S)$.  Note that when a subset of a real vector space
contains $0$, $\aff(S) = \lin(S)$.  The {\em relative interior} of a
convex compact set $C$ is the interior of $C$ when it is considered as
a subset of $\aff(C)$.  Finally, we'll use the term {\em $Z$-ball},
where $Z$ is an affine subspace of the ambient vector space, to mean a
subset of $Z$ that is a ball in $Z$.

The proof relies on the following lemma, proven below.
\begin{lemma}
The mixing entropy fails to be concave for any d-dimensional nonsimplicial
polytope all the maximal faces of which are (d-1)-dimensional simplices.
\end{lemma}

We begin by proving the theorem.

\begin{proof} (Theorem). First note that any counterexample to concavity of
the mixing entropy in a polytope $S$ will also be a counterexample in a
polytope $S^{\prime }$ that has $S$ as a face. \ This follows from the fact
if $S$ is a face, only states in $S$ can appear in convex decompositions of
states in $S$. \ The proof of the theorem is by induction.

Suppose as our induction hypothesis that the mixing entropy fails to
be concave for nonsimplicial polytopes in dimension $d$. \ For every
polytope in dimension $d+1$ either (i) every maximal face is
simplicial, or (ii) there is a maximal face that is
nonsimplicial. \ If case (ii) applies, then there is a face that
constitutes a nonsimplicial polytope of dimension $d$ and by our
induction hypothesis, the mixing entropy fails to be concave for this
face. \ If case (i) applies, then the polytope satisfies the
conditions of the lemma and the mixing entropy fails to be concave by
virtue of the lemma.

To complete the inductive argument we need to show that the mixing entropy
fails to be concave for nonsimplicial polytopes in dimension $d=2$,
the lowest dimension in which there exist nonsimplicial polytopes. \ This
follows from the fact that all of the maximal faces of a 2-dimensional
nonsimplicial polytope are line segments, which are simplices, so that the
conditions of the lemma apply.
\end{proof}

We now prove the lemma.

\begin{proof} (Lemma). Suppose $S$ is a $d$-dimensional polytope that
satisfies the conditions of the lemma, that is, it is nonsimplicial, but all
of its maximal faces are simplicial.
In this case, one can always find two
maximal faces ($(d-1)$-dimensional simplices), $F_{1}$ and $F_{2}$, whose intersection,
$F_{1}\cap F_{2}$, is a $(d-2)$-dimensional simplex.  We define $V_1$ to be the vertex of $F_1$ that is not contained in $F_2 \cap F_2$.  $V_2$ is defined similarly.
Let $\rho _{1}$ be the barycenter of $F_{1},$ $\rho _{2}$ the barycenter of $F_{2}$ and $\rho _{3}$ the barycenter of $F_{1}\cap F_{2}.$
The figures provide examples of pairs of such faces in different dimensions.
\begin{figure}[H]
\centering
\includegraphics[scale=0.6]{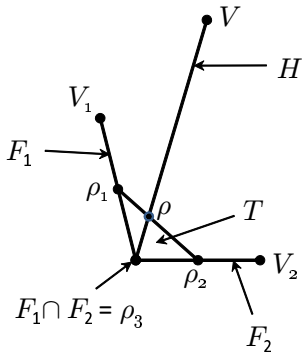}
\caption{Example of failure of concavity
for a 2d nonsimplicial polytope.}
\label{FIG:2dFC}
\end{figure}
\begin{figure}[H]
\centering
\includegraphics[scale=0.5]{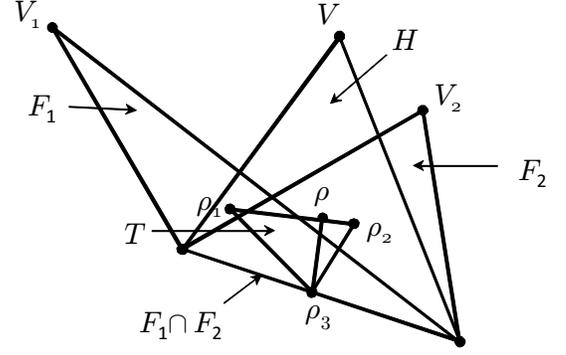}
\caption{Example of failure of concavity
for a 3d nonsimplicial polytope. Here $\protect\rho $ is in the interior of $%
H.$}
\label{FIG:3dFC1}
\end{figure}
\begin{figure}[H]
\centering
\includegraphics[scale=0.5]{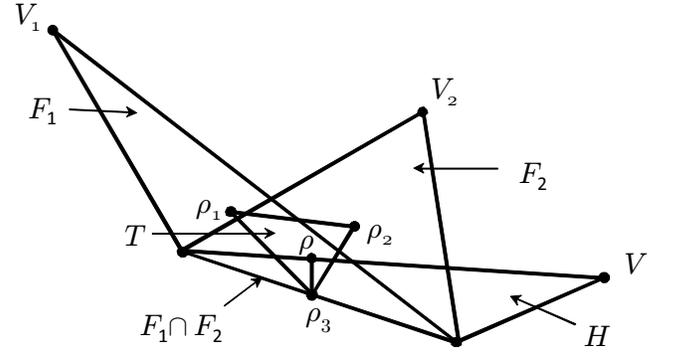}
\caption{Example of failure of concavity for a 3d nonsimplicial polytope.
Here $\protect\rho $ is on the boundary of $H.$}
\label{FIG:3dFC2}
\end{figure}
\black
Let ${V}$ be a vertex of $S$ that is not contained in $F_{1}$ or in $F_{2}.$ \
Such a vertex always exists because if it did not, then the total number of
vertices in $S$ would be $d+1$ and $S$ would be a simplex, contrary to
hypothesis.

Define the $(d-1)$-dimensional polytope $H$ to be the convex hull of
$F_{1}\cap F_{2}$ and ${V}.$ \ Note that $H$ is a simplex.

Define $T$ to be the triangle with vertices $\rho _{1},$ $\rho _{2}$ and $%
\rho _{3}.$

Define $L$ to be the intersection of $T$ and the $(d-1)$-dimensional
polytope $H$.  $L$ is a line segment; we defer establishing this to the end
of the proof, since it is somewhat technical.

Finally, we define the state $\rho$ for which the mixing entropy will fail
to be concave. \ It is defined as the second vertex of $L,$ that is, $L$ is
the line segment extending from $\rho _{3}$ to $\rho.$

Now the proof proceeds differently depending on whether $\rho$ is in the
interior or in the boundary  (relative to $\aff(H)$) of $H$.

i) $\rho $ is in the relative boundary of $H$.

In this case, $\rho $ lies on a face of $H.$ \ Because $H$ is a simplex of
dimension $d-1,$ every such face has $d-1$ vertices and consequently the
mixing entropy of $\rho $ satisfies%
\begin{equation}
S\left( \rho \right) \leq \log (d-1).  \label{eq:LHS1}
\end{equation}%
By definition, $\rho \in T,$ so that it is a convex combination of $\rho
_{1},$ $\rho _{2}$ and $\rho _{3},$%
\begin{equation}
\rho =p_{1}\rho _{1}+p_{2}\rho _{2}+p_{3}\rho _{3},
\end{equation}%
where the $p_{i}$ form a probability distribution. \ Because $\rho _{1}$ $%
(\rho _{2})$ is the barycenter of $F_{1}$ $(F_{2}),$ which has $d$ vertices,
its mixing entropy is%
\begin{equation}
S(\rho _{1})=S(\rho _{2})=\log d.
\end{equation}%
while because $\rho _{3}$ is the barycenter of $F_{1}\cap F_{2}$ with $d-1$
vertices, we have%
\begin{equation}
S\left( \rho _{3}\right) =\log \left( d-1\right) .
\end{equation}%
Recalling that $L$ is $1$-dimensional, we know that $\rho \neq \rho _{3},$
or equivalently, $p_{1}+p_{2}>0,$ which implies that%
\begin{equation}
\sum_{i=1}^3 p_{i}S\left( \rho _{i}\right) >\log \left( d-1\right) .
\label{eq:RHS1}
\end{equation}
\ From Eqs. (\ref{eq:LHS1}) and (\ref{eq:RHS1}) we infer that%
\begin{equation}
S\left( \rho \right) <\sum_{i}p_{i}S\left( \rho _{i}\right) ,
\end{equation}%
that is, the mixing entropy fails to be concave.

ii) $\rho $ is among the relative interior points of $H$.

In this case,%
\begin{equation}
\rho = p_{1}\rho _{1} + p_{2}\rho _{2},  \label{eq:decomprho}
\end{equation}%
that is, $\rho $ lies on the line segment defined by $\rho _{1}$ and $\rho
_{2}.$ \ The proof of (\ref{eq:decomprho}) is by contradiction.  Suppose that $\rho $ lies in the
relative
interior of $T$ as well as in the relative interior of $H.$ \ Then, there
is an $\aff(T)$-ball $B_1$
around $\rho $ contained in the relative interior of $T$ and an
$\aff(H)$-ball around $\rho $ contained in the interior of $H.$ \
$B_1 \intersect B_2$
is a line segment $L_{\rho }\subset L$ with midpoint $\rho$. \ But the fact
that $\rho $ is the midpoint of $L_{\rho }$ contradicts the fact that it is
an extremal point of $L.$

It follows from Eq. (\ref{eq:decomprho}) and the fact that
$\rho_1, \rho_2$ are barycenters of $d-1$-dimensional simplices that
\begin{equation}
\sum_{i}p_{i}S(\rho _{i}) \left(= p_1 S(\rho_1) + p_2 S(\rho_2)\right)
=\log d.  \label{eq:RHS2}
\end{equation}

Next, we show that $\rho $ cannot be the barycenter of $H.$ \ We begin by
demonstrating that $\rho $ is the barycenter of $H^{\prime }\equiv \mathrm{%
conv}\left( F_{1}\cap F_{2},{V}^{\prime }\right) ,$ where ${V}^{\prime }\equiv
p_{1}{V}_{1} + p_{2}{V}_{2}$ and where $\mathrm{conv}\left( S,S^{\prime }\right) $
is the convex hull of $S\cap S^{\prime }.$ \ Letting $\mathrm{bary}\left(
F\right) $ denote the barycenter of $F,$ the proof is as follows, writing
$x_i, i \in \{1,...,d-1\}$, for the $d-1$ vertices of $F_1 \intersect F_2$.%
\begin{eqnarray}
\rho  &=&p_{1}\mathrm{bary}\left( F_{1}\right) +p_{2}\mathrm{bary}\left(
F_{2}\right)  \\
&=& p_1 \frac{1}{d}  (\sum_{i=1}^{d-1} x_i + V_1)
+ p_2 \frac{1}{d}  (\sum_{i=1}^{d-1} x_i + V_2) \\
&=&  \frac{1}{d} ( \sum_{i=1}^{d-1} x_i + p_1 V_1 + p_2 V_2 ) \\
&=&\mathrm{bary}\left( \mathrm{conv}\left( F_{1}\cap
F_{2},p_{1}{V}_{1}+p_{2}{V}_{2}\right) \right)  \\
&=&\mathrm{bary}(H^{\prime }).
\end{eqnarray}%
However, $H^{\prime } = \conv \left(F_1 \intersect F_2, V'\right)$
has a different barycenter from  $H\equiv \mathrm{conv}\left(
F_{1}\cap F_{2},{V} \right) $ because ${V}$ is distinct from ${V}^{\prime }.$ \ The
latter follows from the fact that ${V},$ ${V}_{1}$ and ${V}_{2}$ are all vertices
of the nonsimplicial polytope $S$, and consequently ${V}$ cannot be in
$\conv(F_1 \intersect F_2, V_1, V_2)$, unlike $V'$ which is.

Given that $\rho \in H$ but is not at its barycenter, and given that $H$ has
$d$ vertices, we have%
\begin{equation}
S(\rho )<\log d.  \label{eq:LHS2}
\end{equation}%
From Eqs.~(\ref{eq:RHS2}) and (\ref{eq:LHS2}) we infer the failure of
concavity of the mixing entropy.

We finish the proof by establishing the claim that $L$, defined above, is a line segment.
First note that
because it is an intersection of convex compact sets, it is compact and convex.
Because $\dim(\aff(T))=2$ and $\dim(\aff(H)) = d-1$,
$\aff(T) \intersect \aff(H)$ is one or two dimensional.  For it to be
two dimensional would require $T \subset \aff(H)$, implying $\rho_1,
\rho_2 \in \aff(H)$, and hence since $F_1 \intersect F_2 \subset H$,
that $V_1, V_2$ lie in the hyperplane $\aff(H)$.  That contradicts the
assumption that $F_1$, $F_2$ are distinct maximal faces.

Since $L \subset \aff(T) \intersect \aff(H)$, $L$ is at most one-dimensional.
To show it is \emph{at least }1-dimensional
we begin by observing that because they are subsets of $S$,
both $T$ and $H$ lie in the
``tangent wedge'' $W$ to $S$ at $\rho_3$, i.e. the intersection of the
halfspaces $\aff(F_1)_+$ and $\aff(F_2)_+$.
 Here $\aff(F_i)_+$ is defined
to be the closed half-space to the polytope $S$'s side of $\aff(F_i)$.
In fact, $V$ lies in the
{\em interior} of $W$ because if it lay in $\aff(F_1)$ or $\aff(F_2)$,
our assumption that all maximal faces were simplices would be violated.
Viewing $\rho_3$ as the origin of a real linear space, and noting that
$\lin T$ and $\lin (F_1 \intersect F_2)$ are complementary subspaces (they
span the space and intersect only at $0$, i.e. $\rho_3$) we can decompose
$V$ in a unique way into a component in $\lin(T)$ and a component in the
edge, $\lin(F_1 \intersect F_2)$, of the tangent wedge.

Let $q$ be the linear projection with kernel $\lin(F_1 \intersect F_2)$ and
image $\lin(T)$.  For any set $X$ such that $X = X + \lin(F_1 \intersect F_2)$,
$q(X) = \lin(T) \intersect X$.  Both $W$ and $\aff(H)$ satisfy this condition.
As already noted, $V \in \interior W$; this
is equivalent to $q(V)$ being in the relative interior of $\cone(T)$.   Since
$V \in H$, $q(V)$ is also in $\aff(H)$.
Therefore $\cone(T) \intersect \aff(H)$ is an interior ray $r$ of
$\cone(T)$.
  Now, since $\rho_3$ is the barycenter of $F_1 \intersect F_2$,
there is a $d-2$-dimensional
$\aff(F_1 \intersect F_2)$-ball $B_0$ around $\rho_3$ contained entirely
in $F_1 \intersect F_2$.  Furthermore, $B_0$ is the intersection of a
$d-1$-dimensional $\aff(H)$-ball $B$ around $\rho_3$,
with $F_1 \intersect F_2$.  $\aff(F_1 \intersect F_2)$
divides $\aff(H)$ into halfspaces, and the $H$-side half-ball $B \backslash
B_0$ is contained entirely in the relative interior of $H$.  Since, as established
near the beginning of our argument,
$\aff(T) \intersect \aff(H)$ is a line in $\aff(H)$, and we now know that while
it contains $\rho_3$ it is not
entirely contained in $F_1 \intersect F_2$,
it must intersect $B \backslash B_0$.
Its intersection with $B \backslash B_0$ is contained in $H$.
By choosing $B$ small enough, we can ensure that this intersection
is also contained in $T$.  This is obvious from two-dimensional geometry.
To be slightly more explicit, the facts that $r$, i.e. the half of
$\aff(H) \intersect \aff(T)$ on the $H$-side of
$\aff(F_1 \intersect F_2)$, is interior to $\cone(T)$, $\cone(T)$ is generated by $T$,
and $T$ is closed under multiplication by scalars in $[0,1]$, ensure this.
Since $\aff(T) \intersect \aff(H) \intersect B$ is contained in
both $H$ and $T$, it is contained in $L$; since it is one-dimensional,
so is $L$ and so, since $L$ is a compact convex
set, $L$ is a line segment.
\end{proof}

In generalized theories we can define (cf. also \cite{BKOV03},
where analogous quantities for convex-sets-based theories were
defined, and their failure to be concave in general was also observed)
measurement-entropy-like
quantities $H_T$ based on {\em any} function $T$ that
(like entropy) is Schur-concave
and defined on finite lists of classical probabilities.  For a state
$\rho$, $H_T(\rho)$ is defined as the infimum
over tests of the value of $T$ on
the probabilities for the results of the test.
We define $U_{d}$ for positive integers $d$ as the uniform distribution on
$d$ alternatives.  The same proof as before (with $F(U_d)$ in place of $\log{d}$) gives us:
\begin{proposition}
For any $T$ whose value on $U_{d+1}$ is strictly greater than its value
on $U_{d}$, for all $d$, for example any strictly Schur-concave $T$,
the only polytopes on which $H_T$ is concave are simplices.
\end{proposition}

\section{Entropy and quantum axiomatics}\label{monoentropic}

That mixing and measurement entropies coincide, as they do in
classical and quantum theory, has powerful consequences for the
structure of a probablistic model and, perhaps even more profoundly,
for the structure of a probabilistic {\em theory}. As already noted,
it implies that mixing entropy is concave, which places sharp
restrictions on the geometry of state spaces. It also figures
importantly in our derivation, in Section V, of information
causality. In this Appendix, we explore some further consequences of
monoentropicity, and also suggest some other postulates, the physical
content of which may be clearer, that enforce this property.

It will be helpful to impose some mild restrictions on the models we
consider. (These are satisfied by all of the examples discussed earlier.) First,
we want to have enough analytic structure to guarantee that
measurement entropies will be well-behaved. Accordingly, in this
Appendix we shall require of all models $A = ({\mathfrak A},\Omega)$ that
$\Omega$ be a compact, finite-dimensional convex set, as already
assumed in Section~\ref{testspaces}; additionally, we make the following
technical,
but reasonable and fairly weak, assumptions:
\begin{itemize}
\item[(i)] The total outcome-set $X$ is compact in some Hausdorff topology that makes
every state $\alpha \in \Omega^A$ continuous as a function $\alpha : X \rightarrow [0,1]$.
\item[(ii)] Write $x \perp y$ to mean that outcomes $x$ and $y$ are
 distinct and jointly testable, i.e., there exists a test $E \in
{\mathfrak A}$ containing them both.  We require that $\perp$ be closed
 as a subset of $X \times X$.
 \item[(iii)] ${\mathfrak A}$ is compact in the standard topology it inherits from $X$ (as explained below).
\end{itemize}
 Conditions (i) and (ii) have a certain {\em a priori}
plausibility, and, indeed, are often satisfied in practice: see
\cite{Wilce09} for examples of large classes of test spaces satisfying
them. Condition (iii) requires some further justification. Conditions
(i) and (ii) make $\mathfrak A$ a {\em topological test space}
\cite{Wilce05, Wilce09}.  With $X$ compact, as in condition (i), $\mathfrak
A$ has finite rank \cite{Wilce09}, Lemma 204. This allows us to
topologize the set $\mathfrak A$ of tests as a quotient of a suitable
subspace of $X^{n}$, where $n$ is the rank of $\mathfrak A$. We call this
the {\em standard topology} on $\mathfrak A$. One can show (\cite{Wilce09},
Prop. 211) that $\mathfrak A$ can be enlarged so as to become compact in
this topology, without change to its rank or to its space of
continuous states.  So condition (iii) is relatively harmless. In
fact, if $\mathfrak A$ is {\em uniform}, meaning that all tests have the
same number of outcomes, condition (iii) is automatically satisfied,
given (i) and (ii).

It is not difficult to show that $H_{E}(\alpha)$ is continuous as a
function of $E \in {\mathfrak A}$ (see \cite{Wilce09}, Lemma 210). Hence,
for every state $\alpha \in \Omega$, there exists a test $E$ for which
$H_{E}(\alpha) = H(\alpha)$.  This justifies the assumption to this
effect made in Section~\ref{measurementandmixing}.

We can now characterize those states having zero measurement or mixing
entropy.

\begin{lemma}\label{zeroentropylemma}
Let $\alpha$ be a state of a system $A = ({\mathfrak A},\Omega)$ satisfying the standing assumptions just discussed.
\begin{itemize}
\item[(a)] $H(\alpha) = 0$ iff there exists an outcome $x \in X$ with $\alpha(x) = 1$.
\item[(b)] If $S(\alpha) = 0$, then $\alpha$ is the limit of a sequence of pure states of $\Omega$.
\end{itemize}
\end{lemma}

\begin{proof}
(a) ``if'' follows immediately from the definition, with $E$ any test containing $x$;
``only if'' from the fact (established just above the statement of the Lemma) that $H(\alpha) = H_{E}(\alpha)$ for some $E \in {\mathfrak A}$. For (b), note that if $\vec{p} = (p_1,...,p_n)$ is a discrete probability distribution with $H(\vec{p}) < \epsilon$, then
$\max \{p_i\} > 2^{-\epsilon}$. Now if $S(\alpha) = 0$, we can find, for any sequence $\epsilon_k$ decreasing to $0$, a sequence
of pure-state ensembles $\{p_{i,k} \alpha_{i,k} | i = 1,....,n_k\}$ for $\alpha$ (so that $\alpha = \sum_{i=1}^{n_k} p_{i,k} \epsilon_{i,k}$ for every $k$) with  $H(\vec{p_{k}}) < \epsilon_k$. Ordering each ensemble so that $p_{1,k} = \max\{p_{i,k} | i = 1,...,n_{k}\}$, we find, as above, that
$p_{1,k} > 2^{-\epsilon_{k}}$. Since $\alpha = p_{1,k} \alpha_{1,k} + \sum_{i=2}^{n} p_{i,k} \epsilon_{i,k}$, we have $\alpha > p_{1,k}\alpha_{1,k}$
in the pointwise order on $X = \bigcup {\mathfrak A}$;
consequently, $\|\alpha - p_{1,k}\alpha_{1,k}\| = (\alpha - p_{1,k} \alpha_{1,k})(u) = 1 - p_{1,k}$. Thus, $\alpha = \lim_{k} \alpha_{1,k}$.
\end{proof}

The converse to part (b) would trivially be true if the mixing entropy were continuous on the convex set $\Omega$. However, as Example~\ref{noncontinuousexample} of the main text shows, it need not be.



Call a model $A = (\A,\Omega)$ {\em unital} iff for every  outcome  $x \in X :=
\bigcup {\mathfrak A}$, there exists at least one state $\alpha$ with
$\alpha(x) = 1$. If this state is unique---and therefore pure---for
every $x$, we say that $A$ is {\em sharp}. In this case, we write
$\epsilon_x$ for the unique state with $\epsilon_x(x) = 1$.  In the
literature of quantum axiomatics, sharpness has sometimes been taken
as an axiom (sometimes called {\em Gunson's Axiom})
\cite{Gunson}. Lemma~\ref{zeroentropylemma} has the following corollary:
\begin{corollary}
Suppose $A$ is monoentropic, and that the set of pure states in $A$ is closed. Then $A$ is sharp.
\end{corollary}
\begin{proof}
If $\alpha$ is a pure state, then $H(\alpha) = S(\alpha) = 0$. By
Lemma~\ref{zeroentropylemma}, there exists a measurement outcome $x$
with $\alpha(x) = 1$. On the other hand, if $\alpha(x) = 1$, then
$S(\alpha) = H(\alpha) = 0$, whence, again by
Lemma~\ref{zeroentropylemma}, $\alpha$ is the limit of a sequence of
pure states, say $\epsilon_n \rightarrow \alpha$. By assumption, the
set of pure states is closed, so $\alpha$ is pure. Since the set of
states assigning unit probability to $x$ is convex, it follows that
$\alpha$ is the unique such state.
\end{proof}

While the condition that the set of pure states be closed is not totally innocent (consider, e.g., Example~\ref{noncontinuousexample} above), neither is it unreasonable. For example, it will be satisfied if there exists a compact group of symmetries of the state space that acts transitively on the pure states.

The condition that measurement and mixing entropies coincide also places some constraints on how systems compose:
\begin{lemma}
Suppose that $AB = ({\mathfrak C},\Omega^{AB})$ is a composite (in the
sense of Section~\ref{testspaces}) of systems $A = (\A,\Omega^A)$ and
$B = (\B,\Omega^B)$. Suppose that $A, B$ an $AB$ have closed sets of
pure states, and are monoentropic. If $\Omega^{AB}$ contains an
entangled pure state, then ${\mathfrak C}$ must contain a non-product
outcome.
\end{lemma}
\begin{proof}
By the previous Lemma, $A$, $B$ and $AB$ are sharp. If $x \in \bigcup
{\mathfrak A}$ and $y \in \bigcup {\mathfrak B}$ are outcomes of $A$ and $B$,
respectively, and $\epsilon_x$, $\epsilon_y$ and $\epsilon_{xy}$ are
the unique pure states making $x$, $y$ and $xy$ certain, then
$\epsilon_{xy} = \epsilon_{x} \otimes \epsilon_y$. Now if $\rho$ is a
pure entangled state in $\Omega^{AB}$, then $S(\rho) = 0$. If $H = S$,
then $H(\rho) = 0$, whence, $\rho = \epsilon_{z}$ for some outcome $z
\in \bigcup{\mathfrak C}$. If $z$ is a product outcome, say $z = xy$, then
$\rho = \epsilon_{x} \otimes \epsilon_{y}$ -- a contradiction.
\end{proof}

We now consider whether the condition that $H = S$ can be derived from
more physically transparent considerations.
\begin{definition}
 A probabilistic theory has the {\em pure conditioning property}
 iff, for every pair of systems $A = (\A,\Omega^{A})$ and $B =
 (\B,\Omega^{B})$, every pure state $\omega$ of $AB$, and all
 outcomes $x$ of $A$ and $y$ of $B$, the conditional states $\omega^{B|x}$ and
 $\omega^{A|y}$ are pure. \end{definition}

\begin{lemma}
If a theory satisfies the pure conditioning property, then for any pure bipartite state $\omega$
on a composite $AB$, we have $S(\omega^{B}) \leq H(\omega^{A})$ and
$S(\omega^{A}) \leq H(\omega^{B})$.\end{lemma}
\begin{proof}
Let $\omega$ be a pure bipartite state. Pick an observable $E$
minimizing measurement entropy for $\omega_1$, so that $H(\omega_1)$
is the Shannon entropy $H_{E}(\omega^{A}) := -\sum_{x \in E}
\omega^{A}(x) \log(\omega^{A}(x))$.  We have $\omega^{B} = \sum_{x \in
  E} \omega^{A}(x) \omega^{B|x}$. By PC (and the assumption that
$\omega$ is pure), the conditional states $\omega^{B|x}$ are pure. By
definition, $S(\omega^{B})$ is the minimum Shannon entropy of the
mixing coefficients in any pure-state ensemble for $\omega^{B}$, so
$S(\omega^{B}) \leq H_{E}(\omega^{A}) = H(\omega^{A})$. By the
same argument, $S(\omega^{A}) \leq H(\omega^{B})$.  \end{proof}

\begin{definition}
A theory has the {\em steering property} iff, for every pair of systems $A$ and $B$,
every pure bipartite state $\omega$ of $AB$ {\em steers} its
marginals, in the sense that for any convex decomposition $\omega^{B}
= \sum_{i} p_i \beta_i$, with $\beta_i$ pure and distinct from each other,
there is a test
$E = \{a_i\}$ of $A$ with $\beta_i = \omega^{B|a_i}$, and
similarly for $\omega^{A}$.
\end{definition}

The term ``steering'' is due to Schr\"odinger \cite{Schrodinger}, who
showed that quantum theory
is steering; further proofs and extensions  are in  Hadjisavvas \cite{Hadjisavvas}
and Hughston, Jozsa, and Wootters \cite{Hughston}; a survey is \cite{Kirkpatrick}.

\begin{lemma} If a theory has the steering property, then for every pure bipartite state $\omega$, $H(\omega^{A}) \leq S(\omega^{B})$.\end{lemma}
\begin{proof}
For any $\epsilon > 0$, choose a convex decomposition $\omega^{B} = \sum_{i} p_i \beta_i$ of $\omega^{B}$ into pure states $\beta_i$, with
$S(\omega^{B}) > H(p_i) - \epsilon$.
Since the state $\omega$ is steering, there exists a test $E = \{x_i\}$ with $\omega^{B|x_i} = \beta_i$, whence $p_i = \omega^A (x_i)$.
It follows that $S(\omega^{B}) > -\sum_i p_i \log(p_i) - \epsilon = H_{E}(\omega^{A}) - \epsilon$.
Since $\epsilon$ is arbitrary, $S(\omega^{B}) \geq H(\omega^{A})$.
\end{proof}

\begin{definition}
A pure state $\alpha$ in an abstract probabilistic
theory is {\em purifiable} iff for every state $\alpha$ on a system $A$, there
exists a pure bipartite state $\omega$ -- a {\em purification} of
$\alpha$ -- on a composite $AB$, with $B$ a copy of $A$, with
$\omega^{A} = \omega^{B} = \alpha$.
An abstract probabilistic theory has the {\em purifiability property} iff
every state in the theory is purifiable.
\end{definition}
Quantum mechanics has the purifiability property.  D'Ariano {\em et al.} \cite{D'Ariano09} have considered a
condition very similar to purifiability as a potential axiom for
quantum theory, and have shown that many other features of quantum
theory follow from it. From the Lemmas above, we have

\begin{proposition}  A theory that has the pure conditioning, steering
and purifiability properties is  monoentropic.  \end{proposition}

\section{Linearized test space models, ordered linear space models, and entropy}
\label{appendix:linearization}

The apparatus of states on test spaces can be linearized, as
follows. If $A = ({\mathfrak A},\Omega)$, with total outcome space $X =
\bigcup {\mathfrak A}$, let $V(A)$ denote the span of $\Omega$ in ${\mathbb
  R}^{X}$, regarded as an ordered real vector space with positive cone
$V_{+}(A) = \{ \alpha \in V(A) | \alpha(x) \geq 0 \forall x \in
X\}$. Every outcome $x \in X$ defines a positive linear evaluation
functional $f_x \in V^{\ast}(\Omega)$ by $f_x(\mu) = \mu(x)$ for all
$\mu \in V(A)$. Moreover, one has $\sum_{x \in E} f_{x} = u$, where
$u$ is the unique functional taking the constant value $1$ on
$\Omega$. Abstracting, one defines an {\em effect} to be a positive
linear functional $f \in V^{\ast}(A)$ with $0 \leq f(\alpha) \leq 1$
for all $\alpha \in \Omega$ (equivalently, $0 \leq f \leq u$); an {\em
  observable} on $A$ is a sequence $f_1,...,f_n$ of effects with
$\sum_{i} f_i = u$.

From this point of view, the structure of the test space is
essentially a privileged set of observables -- an additional structure
that (like a preferred basis for a vector space) may or may carry some
useful information, or may simply be a computational convenience. For
example, if $A(\H) = ({\mathfrak F}(\H),\Omega(\H))$ is a quantum
system, $V(A)$ is the space of quadratic forms associated with -- but
one might as well say, the space {\em of} -- Hermitian operators on
$\H$, and $V^{\ast}$ is essentially the same space, under the duality
$a(\rho) = \Tr(\rho a)$.  In particular, an effect is a positive
operator between $0$ and $\mathbf{1}$, and an observable is
essentially a discrete POVM.  The convex sets, or ordered linear
spaces, formalism takes this kind of combination of a convex state space
and a set of effects in the dual cone to the state space, as
primary.  Roughly, a convex model is defined by taking a convex
compact set of states as a base for a cone $V(\Omega)_+$ of
unnormalized states, and a cone of ``unnormalized allowed effects''
that is a closed subcone $V^{\sharp}_+$, containing $u$ in its interior, of
the dual cone $V^*(\Omega)_+$ of all effects.  $u$ is defined by the
condition $u(\Omega) = 1$, and the interval $[0,u]$ according to the
ordering defined by $V^\sharp$ is the set of effects allowed in the
theory. When $V^\sharp = V(\Omega)_+$, the model is called {\em maximal}
(or sometimes \emph{saturated} \cite{BW09b}).  If the model is constructed from a
test space, one will usually want to choose $V^\sharp$ to contain the
effects associated with all outcomes in the test space.

Two natural distinguished classes of effects are the
\emph{ray-extremal} ones, that is effects that lie on extremal rays of
the cone generated by effects, and \emph{atomic} effects, i.e.,
maximal effects in extremal rays (equivalently, ray-extremal effects
that are also extremal in the convex set $[0, u]$ of effects).  We may
define the measurement entropy as the infimum of entropies obtainable
by measuring observables consisting of ray-extremal elements, or
alternatively as the infimum of entropies obtainable by measuring
observables consisting of atomic effects.  Intuitively, the
observables consisting of \emph{ray-extremal} effects are maximally
fine-grained.  Ray-extremal effects cannot be further refined by
decomposing them as sums of other effects.  Although they can be
decomposed as sums of shrunken versions of themselves, intuitively
this cannot provide any additional information about the system being
measured.  Certainly in the case of atomic effects, and probably with
some care and relabeling in the case of ray-extremal effects (which
unlike atomic effects may appear more than once in a given
observable), the measurements with such outcomes can be organized into
distinguished test spaces associated with a given convex-sets model,
so the test space framework we use in the main text will probably
cover this natural possibility, although the additional assumptions we
make to obtain particular results will need to be checked for these
cases.  In the case of ray-extremal effects, the infimum in the
definition of measurement entropy will likely not be changed if we
omit measurements in which an effect appears more than once; the
measurements without repetitions should be easier to organize into a
test space.  Linearization and the application of one of these
definitions may well remove pathologies in measurement entropy that
are associated with some test space/state space models.  The spirit of
the definition of measurement entropy via an infimum suggests
excluding tests that are not maximally fine grained when
viewed from the convex states perspective, as the above definitions
do.  Passing to these definitions may also remove pathologies
that might arise when the set of distinguished observables associated
with tests has an irregular relationship to a state space whose
underlying geometry is quite regular.

\end{document}